\documentclass[twocolumn,10pt,notitlepage,nofootinbib,superscriptaddress,numerical]{revtex4-2} 

\usepackage{amsmath,amssymb}
\usepackage{dsfont}
\usepackage{xcolor}
\usepackage{amsthm}
\usepackage[colorlinks=true, pdfstartview=FitV, linkcolor=blue, citecolor=blue, urlcolor=blue]{hyperref}
\usepackage{array}

\newtheorem{theorem}{Theorem}[section]

\newtheorem{corollary}[theorem]{Corollary}
\newtheorem{prop}[theorem]{Proposition}
\newtheorem{lemma}[theorem]{Lemma}

\theoremstyle{definition}
\newtheorem{definition}[theorem]{Definition}
\newtheorem{example}{Example}[section]
\newtheorem{remark}{Remark}

\usepackage{graphicx}
\usepackage[caption=false]{subfig}
\graphicspath{
    {figures/}
}

\usepackage{mathtools} 
\usepackage[shortlabels]{enumitem}
\usepackage{bm}
\usepackage{bbm}
\usepackage{dsfont} 
\usepackage{etoolbox}

\docsvlist{A,B,C,D,E,F,G,H,I,J,K,L,M,N,O,P,Q,R,S,T,U,V,W,X,Y,Z}

\docsvlist{A,B,C,D,E,F,G,H,I,J,K,L,M,N,O,P,Q,R,S,T,U,V,W,X,Y,Z}

\docsvlist{A,B,C,D,E,F,G,H,I,J,K,L,M,N,O,P,Q,R,S,T,U,V,W,X,Y,Z}

\docsvlist{d}

\docsvlist{E,P,R,e,k,p,q,r,g,x,y}

\usepackage{ stmaryrd }
\def\llb{\llbracket}
\def\rrb{\rrbracket}

\def\ket{\rangle}
\def\lr{\leftrightarrow}

\def\Weight{\mathfrak{W}}
\def\weight{\mathfrak{w}}
\def\qubit{\mathfrak{q}}

\usepackage{algorithm}
\usepackage{algpseudocode}

\algnewcommand\algorithmicforeach{\textbf{for each}}
\algdef{S}[FOR]{ForEach}[1]{\algorithmicforeach\ #1\ \algorithmicdo}

\DeclareMathOperator{\im}{im}

\DeclareMathOperator{\cone}{cone}
\DeclareMathOperator{\coker}{coker}

\usepackage{tikz-cd} 
\usetikzlibrary{matrix,arrows,decorations.pathmorphing}
\usetikzlibrary{fit, shapes.geometric,calc}

\usepackage{pifont}

\usepackage{xcolor}

\begin{document}

\onecolumngrid

\clearpage

\title{Parsimonious Quantum Low-Density Parity-Check Code Surgery}
\author{Andrew C.~Yuan}
\affiliation{Iceberg Quantum, Sydney}
\affiliation{Condensed Matter Theory Center and Joint Quantum Institute,
Department of Physics, University of Maryland, College Park, Maryland 20742, USA}
\author{Alexander Cowtan}
\noaffiliation
\author{Zhiyang He}
\affiliation{Department of Mathematics, Massachusetts Institute of Technology, Cambridge, MA 02139, USA}
\author{Ting-Chun Lin}
\affiliation{Department of Physics, University of California San Diego, CA}
\author{Dominic J.~Williamson}
\affiliation{School of Physics, The University of Sydney, Sydney, NSW 2006, Australia}
\date{March 2026}

\begin{abstract}
    Quantum code surgery offers a flexible, low-overhead framework for executing logical measurements within quantum error-correcting codes. 
    It encompasses several fault-tolerant logical computation schemes, including parallel surgery, universal adapters and fast surgery, and serves as the key primitive in extractor architectures. 
    The efficiency of these schemes crucially depends on constructing low-overhead ancilla systems for measuring arbitrary logical operators in general quantum Low-Density Parity-Check (qLDPC) codes. 
    In this work, we introduce a method to construct an ancilla system of qubit size $O(\Weight \log \Weight)$ to measure an arbitrary logical Pauli operator of weight $\Weight$ in any qLDPC stabilizer code.
    This new construction immediately reduces the asymptotic overhead across various quantum code surgery schemes.
\end{abstract}

\maketitle

\section{Introduction}
\label{sec:intro}
\begin{figure}[ht]
\centering
\subfloat[]{%
    \centering
    \includegraphics[width=0.4\columnwidth]{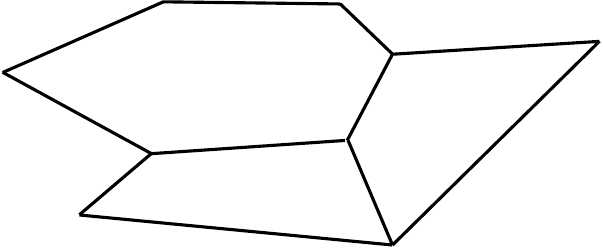}
}
\subfloat[\label{fig:decongestion-ancilla}]{%
    \centering
    \includegraphics[width=0.6\columnwidth]{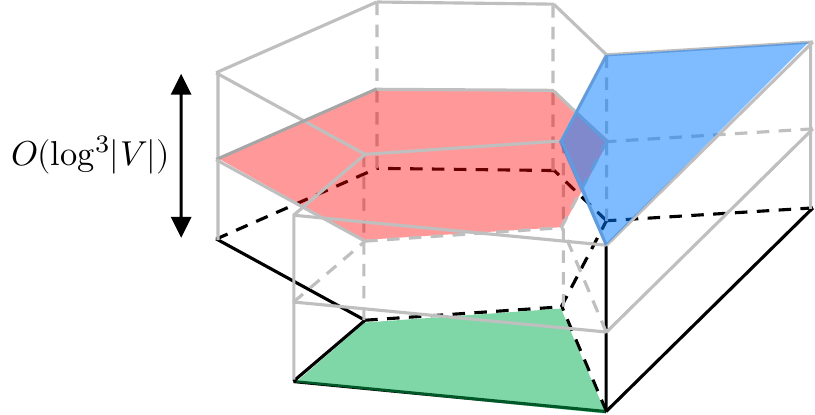}
}
\\
\subfloat[\label{fig:parsimonious-ancilla}]{%
    \centering
    \includegraphics[width=.7\columnwidth]{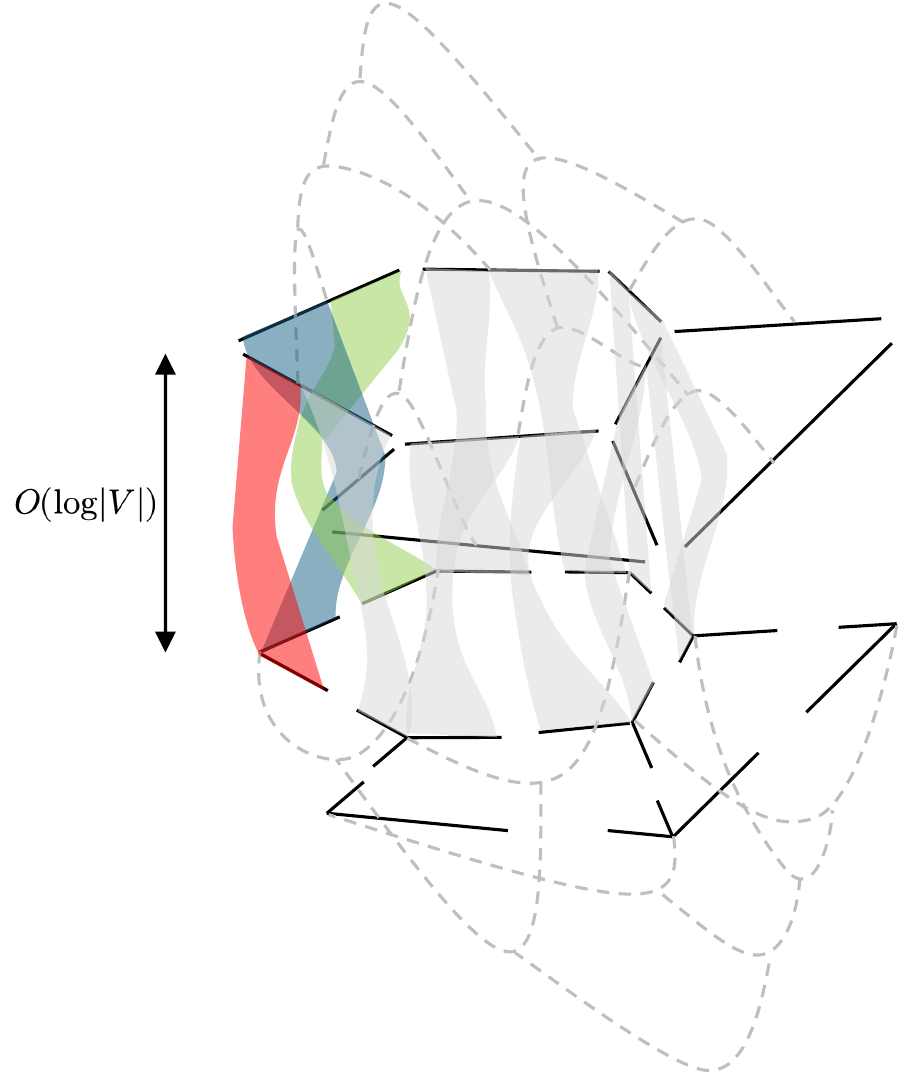}
}
\caption{Parsimonious Ancilla System. (a) An example of an induced graph of a logical operator. (b) The ancilla system obtained via the decongestion lemma + thickening with height $O(\log^3 |V|)$ such that cycles are now contractible due to the faces (colored) and the congestion per edge is controlled to remain constant. (c) The ancilla system obtained via the parsimonious cone. Specifically, two binary trees are constructed, based on the \textit{local structure} of edges (top) and vertices (bottom). \textit{Shuffles} of $O(\log |V|)$ binary trees induce faces (depicted in color and grey) so that the local views are paired up to reveal the \textit{global structure}. {Specifically, any cycle in (a) is embedded in (c) by traversing between the top and bottom binary trees so that the boundaries of the shuffling faces correspond to a cycle in (a) (e.g., colored and grey). Hence, cycles in (a) are products of local, low-density faces in (c).} 
}
\label{fig:cartoon-comparison}
\end{figure}


Quantum error correction forms the essential foundations for reliable, large-scale quantum computing~\cite{shor1995scheme, gottesman1997stabilizer}. 
While the surface code~\cite{bravyi1998quantum, kitaev2003fault} has been the primary focus for decades due to its high thresholds and 2D connectivity, its low encoding rate presents a significant challenge for scaling~\cite{bravyi2010tradeoffs}.
This has led to the development of quantum Low-Density Parity-Check (qLDPC) codes~\cite{breuckmann2021quantum}, which can achieve much higher encoding rates and distances.
Recent breakthroughs~\cite{hastings2021fiber,breuckmann2021balanced,panteleev2021degenerate} have culminated in the discovery of ``good" qLDPC codes~\cite{panteleev2022asymptotically}, in which the quantum dimension and distance scale linearly with the system size, and high-rate qLDPC codes of finite length with high pseudo-thresholds against circuit-level noise~\cite{bravyi2024high,xu2024constant-overhead}.

Although qLDPC codes provide efficient quantum memory, their high rate of encoded logical qubits makes it difficult to perform fault-tolerant computation. 
In high-rate codes, logical qubits are global degrees of freedom rather than localized patches of physical qubits. 
This makes it hard to address and perform computation on individual logical qubits without affecting the rest of the block. 
Standard methods like (fold)-transversal gates~\cite{breuckmann2024fold,eberhardt2024logical} and automorphisms~\cite{calderbank1997quantum, sayginel2025fault} are restricted by no-go theorems~\cite{eastin2009restrictions, chakraborty2026nogo}, which limit their utility for universal fault-tolerant quantum computation.

Quantum code surgery offers a flexible, code-agnostic method for performing logical operations. 
By using an auxiliary ancilla system to temporarily modify the code's stabilizers, surgery allows for the measurement of logical Pauli operators~\cite{cohen2022low, cowtan2024css, cross2024improved,williamson2024low,ide2025fault} and more general operators~\cite{davydova2025universalfaulttolerantquantum,huang2026hybridlatticesurgerynonclifford,zhu2026nonabelianqldpctqftformalism,williamson2026fastmagicstatepreparation,christos2026nonabelianquantumlowdensityparity}. 
This technique, based on principles of lattice surgery~\cite{horsman2012surface} and weight reduction~\cite{hastings2021quantum}, provides a path to universal fault tolerance when combined with magic state injection~\cite{bravyi2016trading,litinski2019game}, and serves as the core primitive in fault-tolerant qLDPC architectures that fall under the umbrella of extractor architectures~\cite{he2025extractors,yoder2025tour,webster2026pinnacle}.

A primary concern in code surgery is the ancilla overhead. 
Since logical operators in qLDPC codes can have high weights, the number of extra qubits needed to perform measurements can be substantial. 
If this overhead is too high, it undermines the physical qubit savings that qLDPC codes are intended to provide.
Fortunately, it was shown that an ancilla system with size $O(\Weight \log^3 \Weight)$ can be constructed for any logical operator of weight $\Weight$ in a general qLDPC stabilizer code~\cite{williamson2024low,he2025extractors}.
This construction, depicted in Fig. \ref{fig:decongestion-ancilla}, was obtained by modifying the weight reduction method pioneered by Hastings \cite{hastings2021quantum} -- specifically, the decongestion and thickening lemmas~\cite{freedman2021building}.
The ancilla system for the gauging logical measurement procedure~\cite{williamson2024low}, which measures a single logical operator, has since been used in the construction of several other surgery primitives, such as universal adapters~\cite{swaroop2026universal}, parallel logical measurement~\cite{cowtan2025parallel},  extractors~\cite{he2025extractors}, as well as the single-shot surgery gadgets of Ref.~\cite{baspin2025fast} and the constant-time surgery gadgets of Ref.~\cite{chang2026constant}.\footnote{We refer readers to Section~IA of Ref.~\cite{chang2026constant} for a brief discussion of the two different formulations of fast quantum surgery, single-shot surgery~\cite{hillmann2025single,tan2025single,baspin2025fast} and constant-time surgery~\cite{cowtan2025fast,chang2026constant}.} 
The spatial overheads of these constructions are tabulated in Table.~\ref{tab:tabulated-improvement}.

\begin{table*}
    \centering
    \begin{tabular}{ |c||c|c| }
    \hline
    Space overhead & Previous best & This work \\
    \hline
    Single measurement & $\mathcal{O}(\Weight \log^3 \Weight)$~\cite{williamson2024low}  &  $\mathcal{O}(\Weight \log \Weight)$\\
    Universal adapters & $\mathcal{O}(t\Weight\log^3\Weight)$~\cite{swaroop2026universal} & $\mathcal{O}(t\Weight \log \Weight)$ \\
    Parallel measurement & $\mathcal{O}(t\Weight (\log t + \log^3\Weight))$~\cite{cowtan2025parallel} & $\mathcal{O}(t\Weight(\log t + \log \Weight))$ \\
    Extractors & $\mathcal{O}(n \log^3 n)$~\cite{he2025extractors} & $\mathcal{O}(n \log n)$ \\
    Single-shot surgery & $\mathcal{O}(nkd(\log k + \log^3 n))$~\cite{baspin2025fast} & $\mathcal{O}(nkd(\log k + \log n))$\\
    Constant-time surgery & $\mathcal{O}(n \log^3 n)$~\cite{chang2026constant} & $\mathcal{O}(n \log n)$ \\
    \hline
    \end{tabular}
    \caption{Comparison of the asymptotic space overheads of different fault-tolerant logical measurement schemes on LDPC codes, prior to this work and after including the construction in this work. In the table, $t$ refers to the number of logical Pauli operators being measured, $\Weight$ is the maximum weight of the logical representatives being measured, $n$ is the blocklength, $k$ is the quantum dimension and $d$ is the code distance. 
    For adapters, it is presumed that $t$ different sparsely overlapping logical operators are measured jointly. 
    The first four gadgets are applicable to all stabilizer codes. 
    The single-shot surgery gadget from Ref.~\cite{baspin2025fast} is applicable to codes that admit single-shot syndrome extraction.
    The constant-time surgery gadget from Ref.~\cite{chang2026constant}, based on work of Ref.~\cite{cowtan2025fast}, is applicable to hypergraph product codes.
    }
    \label{tab:tabulated-improvement}
\end{table*}

In this work, we introduce a method for constructing single-operator ancilla systems with a guaranteed asymptotic overhead that is lower than previous constructions, i.e., ancilla systems of size $O(\Weight\log \Weight)$.
The improved method, depicted in Fig. \ref{fig:parsimonious-ancilla},  relies on the recently introduced quantum weight reduction procedure in Ref.~\cite{hsieh2025simplified} (which is based on the techniques in  Ref.~\cite{berdnikov2022parsimonious}). 
This improved ancilla system has the immediate knock-on effect of reducing the asymptotic space overhead for a variety of surgery schemes which are tabulated in Table~\ref{tab:tabulated-improvement}. 
We expect the $O(\Weight\log \Weight)$ upper bound on the overhead of our scheme for measuring a single operator of weight $\Weight$  to be the optimal result achievable by any general quantum code surgery procedure that is based on an auxiliary graph, which we briefly review in Sec.~\ref{sec:measurement-graph}. 
Due to the technical obscurity of Refs.~\cite{hsieh2025simplified, berdnikov2022parsimonious}, we provide detailed exposition and proofs of various claims, with our main result stated as Theorem \ref{thm:parsimonious-ancilla} (or more specifically, Theorem \ref{thm:cellulated-cone}).\footnote{As an aside, though irrelevant to the ancilla overhead, our analysis in Remark~\ref{rem:minor-difference} shows that a minor issue was not accounted for in the original reference~\cite{hsieh2025simplified}, which may result in weaker guarantees on the constants for reduced check weights and total qubit degree.}

The remaining sections are laid out as follows. 
In Section~\ref{sec:intuitive} we provide a high level overview of our construction.
In Section~\ref{sec:Prelim} we introduce relevant background material.
In Section~\ref{sec:cone} we describe an explicit construction of an ancilla system based on a parsimonious cone.
In Section~\ref{sec:attaching-ancilla-cone} we discuss how an ancilla system based on a parsimonious cone can be attached to a qLDPC code.
In Section~\ref{sec:discuss} we discuss our results and potential future directions.

\section{Overview of Construction}
\label{sec:intuitive}


The main goal of this work is to apply the quantum weight reduction procedure recently introduced in Ref.~\cite{hsieh2025simplified} to quantum code surgery. 
In pursuit of this goal, we provide a detailed exposition of the ideas and techniques in Ref.~\cite{hsieh2025simplified}.
We focus on the setting of measuring logical operators,
where the overhead scaling is dominated by the decongestion and thickening procedure introduced by Freedman and Hastings~\cite{freedman2021building} which has been used widely~\cite{hastings2021quantum,williamson2024low, swaroop2026universal,cowtan2025parallel,he2025extractors,baspin2025fast,chang2026constant}. 
Here, we describe a more efficient coning method that replaces the decongestion and thickening procedure and hence leads to more efficient quantum code surgery. 
For concreteness, in the main text we restrict our attention to a single logical measurement of an $X$-type or $Z$-type operator on a Calderbank-Shor-Steane (CSS) code~\cite{calderbank1996good}. 
The extension to other schemes, including to arbitrary Pauli operators on general stabilizer codes, and to extractor architectures, is presented in Appendix~\ref{app:other_schemes}.

\subsection{Measurement graph}
\label{sec:measurement-graph}

Direct measurement of an $X$-type logical operator on a qLDPC code $D$ is not fault tolerant in general, assuming the code has a large code distance. 
Specifically, let $X(\ell)$, for $\ell\in \dF_2^n$, denote a logical operator specified by a product of Pauli $X$ operators on qubits with $\ell_i=1$, and $\Weight\equiv |\ell|\ge d(D) \gg 1,$ denote the weight of this operator. 
A potentially fault-tolerant method of measuring $X(\ell)$ is to introduce an ancilla code $A$ such that by \textit{merging} $A$ and $D$, the output code $C$ now contains $\ell$ as a stabilizer check while remaining LDPC with large code distance $d(C)\gg 1$.
Hence, the logical measurement of $\ell$ can be inferred from $O\big(d(C)\big)$ rounds of syndrome extraction of the code $C$.

When constructing the ancilla system, the common first step is to consider the measurement graph $\sG$ of $\ell$, defined below.


\begin{definition}[Measurement Graph]
    \label{def:measurement-graph}
    The \textbf{measurement graph} $\sG$ of $\ell$ has vertices and edges defined as follows. 
    Let $\sV$ denote the collection of qubits on which $\ell$ acts non-trivially.
    For every $Z$-check $z$ of $D$ which acts nontrivially on qubits in $\ell$, the common overlap of qubits must be even in cardinality since they commute, and thus can be arbitrarily paired up $qq'$. 
    The collection of all pairs forms the edges $\sE$.

   To preserve the deformed code distance, arbitrary edges $\sE'$ are added to the graph while keeping the max degree fixed so that the measurement graph is connected and expanding with $\Omega(1)$ Cheeger constant.
\end{definition}

The rationale behind the choice of $\sG$ is that the unique connected component of $\sG$ is exactly equal to $\ell$, while $\sG$ itself is a graph with bounded degree.
Specifically, we consider an ancilla system $G$ constructed based on $\sG$ so that there is one qubit per edge and an $X$-check per vertex acting on adjacent edges, i.e., $\prod_{e\sim v} X^{G}_e$.
\textit{Merge} the ancilla $G$ with the data code $D$ so that ancilla $X$-checks also acts on the data qubit of the corresponding vertex $v$, i.e., $X_v^{D} \prod_{e\sim v} X^{G}_e$.
The $Z$-checks of $D$ which anticommute with the new ancilla $X$-checks can be deformed onto the ancilla edge qubits in the graph, and we obtain a new stabilizer code $C$.
In this case, each ancilla $X$-check has bounded weight, while the product of all ancilla $X$-checks is $=X^D(\ell)$, and thus the intuition is to infer the logical measurement of $X^D(\ell)$ via measurements of bounded weight ancilla $X$-checks. 

Unfortunately, this construction does not consider the possibility that the ancilla system has internal logical operators.
For example, before \textit{merging}, the ancilla $G$ may have nontrival $Z$-type logical operators which act on cycles in $\sG$.
Since the weight of these cycles are not well-controlled, they may cause the code distance of the output code $C$ to be small once the ancilla and data codes are merged together, and thus syndrome extraction of the output code $C$ is no longer guaranteed to be fault-tolerant.

\subsection{Coning}

\begin{figure}[ht]
\centering
\subfloat[\label{fig:graph}]{%
    \centering
    \includegraphics[width=0.45\columnwidth]{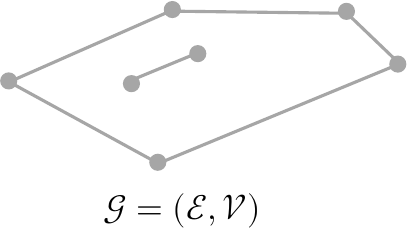}
}
\subfloat[\label{fig:cone-graph}]{%
    \centering
    \includegraphics[width=0.45\columnwidth]{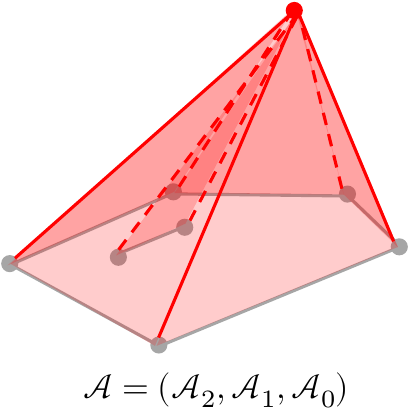}
}
\caption{Coning. (a) depicts the measurement graph $\sG$ of an $X$-type logical with vertices and edges depicted by dots and lines, respectively.
(b) depicts the cone of $\sG$, which is a cell complex obtained by adding an extra vertex (red), and connecting the vertex to all vertices in the original graph. All vertical faces (shaded red) are also added to the cell complex.
}
\label{fig:cone-replacement}
\end{figure}

A natural extension of the approach described in the previous subsection involves introducing bounded weight faces $\sF$ to the graph $\sG$ so that the cycles are generated by the boundaries, which ensures that the constructed ancilla code will possess no internal logicals.
This was first proposed in Ref.~\cite{cross2024improved}, without a guarantee on the weight and degree of the faces, and later refined in Ref.~\cite{williamson2024low},  providing an LDPC guarantee on the faces~\cite{freedman2021building}.
A topological method that achieves this goal is known as \textit{coning}.
This is illustrated in Fig. \ref{fig:cone-replacement}, where a single vertex $\star$ is added to a graph, and further edges are obtained by connecting the vertices of $\sG$ with $\star$, while faces are obtained by connecting the edges of $\sG$ with $\star$.
In this case, $\sG$ is a subgraph of $\cone \sG$. 
The cycles of $\sG$ are \textit{contractible} since they can be generated by the faces of $\cone \sG$ (e.g., the unique cycle in Fig. \ref{fig:graph} is the boundary of all vertical faces in Fig. \ref{fig:cone-graph}), and the cone has a unique connected component since all vertices are connected to $\star$. 

Coning introduces new problems, mainly the $\star$ vertex is not bounded in degree, and thus sparsification is necessary.
This was initially achieved using the decongestion lemma by Freedman and Hastings \cite{freedman2021building, hastings2021quantum}, in which an LDPC ancilla system of qubit size $O(\Weight \log^3 \Weight)$ is constructed.
Compared to the naive ancilla system $G$ with qubit size $O(\Weight)$, the decongestion procedure only incurs a polylogarithmic overhead, and thus is commonly used in the quantum weight reduction and logical measurement literature.
However, Ref.~\cite{hsieh2025simplified} applies techniques from Ref.~\cite{berdnikov2022parsimonious} to obtain an asymptotically lower ancilla size of $O(\Weight \log \Weight)$, while following the \textit{coning} philosophy.
This is captured in the following theorem. 

\begin{theorem}[Parsimonious Cone]
    \label{thm:parsimonious-ancilla}
    Let $\sG=(\sE,\sV)$ denote a graph with bounded degree.
    Then there exists cell complex $\sA$ with faces $\sA_2$, edges $\sA_1$, vertices $\sA_0$ such that 
    \begin{enumerate}
        \item All faces are adjacent to bounded number of edges (and vice-versa);
        \item All cycles of $\sA$ are generated by boundaries of faces in $\sA_2$;
        \item $(\sA_1,\sA_0)$ is a connected graph of bounded degree containing $\sG$ as a subgraph;
        \item $|\sA_0| = O(|\sV| \log|\sV|).$
    \end{enumerate}
\end{theorem}
See Theorem \ref{thm:cellulated-cone} for full statement and proof.
The ancilla LDPC CSS code $A$ is constructed so that there is a qubit per edge in $\sA_1$, and an $X$- ($Z$-) check per vertex $\sA_0$ (per face $\sA_2$) acting on adjacent qubits.
This ancilla system can then be used for measuring the logical operator $X^D(\ell)$ fault-tolerantly.






\subsection{Parsimonious Measurement}

To perform a fault-tolerant logical measuremnt, an ancilla system is constructed by applying Theorem \ref{thm:parsimonious-ancilla} to an appropriate sparse expander graph $\mathcal{G}$, and using it to design an ancilla system following the standard approach in literature, e.g., Ref.~\cite{williamson2024low}.
For the sake of completeness, we include the details of this procedure in Section \ref{sec:attaching-ancilla-cone}, whose main result is presented here.

\begin{theorem}[Deformed Code]
    \label{thm:deformed}
    Let $D$ be a $\llb n,k,d\rrb$ data LDPC code with a chosen logical $X$-type operator $X^D(\ell)$  of weight $\Weight$. 
    Let $A$ be the ancilla LDPC CSS code constructed from the parsimonious cone $\sA$ in Theorem \ref{thm:parsimonious-ancilla}.
    Then there exists a method to construct a deformed $\llb n+O(\Weight \log \Weight), k-1, \Omega(d)\rrb$ LDPC code $C$ from ancilla $A$ and data $D$ such that $X^{D}(\ell)$ is a stabilizer of $C$.
\end{theorem}

\begin{proof}
    See Section \ref{sec:attaching-ancilla-cone} or specifically, Theorem \ref{thm:deformed-code}
\end{proof}

\section{Preliminaries}
\label{sec:Prelim}

In this section we introduce relevant background material that is used throughout the remainder of the work. 

\subsection{Complexes}
\begin{definition}
\label{def:complex}
A \textbf{complex} $C$ is a sequence of (finite-dimensional) $\dF_2$-vector spaces $C_{i}$ together with $\dF_2$-linear maps\footnote{Subscript $i$ may be omitted.} $\partial_{i}:C_{i} \to C_{i-1}$, called the \textbf{differentials} of $C$, such that $\partial_{i}\partial_{i+1}:C_{i+1} \to C_{i-1}$ is zero. We write\footnote{We adopt convention that $C$ ends at $\cdots \to C_0$.}
\begin{equation}
    C = \cdots \to  C_{i} \xrightarrow{\partial_{i}} C_{i-1} \to \cdots
\end{equation}
Note that $\im \partial_{i+1} \subseteq \ker \partial_i$ for all $i$, and thus the \textbf{$i$-homology} of $C$ is defined as
\begin{equation}
    H_i(C) = \ker \partial_i/\im \partial_{i+1}
\end{equation}
Denote the equivalence class of $\ell \in C_{i}$ as $[\ell]\in H_i(C)$ so that the notation $[\cdots]$ is reserved for the quotient map -- conversely, we may also write $[\ell] \in H_i(C)$ without specifying the representation $\ell$.
The complex $C$ is \textbf{exact} if $H_i(C) =0$ for all $i$.
\end{definition}

\begin{definition}[Basis]
\label{def:basis}
A complex $C$ is equipped with an \textbf{basis} (which we always assume henceforth) if each degree $C_i = \dF_2^{n_i}$ with basis elements referred as \textbf{$i$-cells}. 
The  \textbf{(Hamming) weight} $|\ell|$ is then defined for any $\ell \in C_{i}$.
\end{definition}



\begin{definition}[Cocomplex]
Let $C$ denote a complex.
Since the inner product is nondegenerate, the transpose $\partial_i^\top:C_{i-1} \to C_i$ is well-defined so that that \textbf{cocomplex} is that given by
\begin{equation}
    C^\top = \cdots \leftarrow C_{i} \xleftarrow{\partial_{i}^\top} C_{i-1} \leftarrow \cdots
\end{equation}
The \textbf{$i$-cohomology} of $C$ is $H^i(C) = \ker \partial_{i+1}^\top/\im \partial_i^\top$ which is isomorphic to $H_i(C)$.
\end{definition}

\begin{example}[CSS Codes]
\label{ex:CSS}
Consider a CSS code with parity check matrices $H_X,H_Z$. Then the code can be treated as a length-2 complex, i.e.,
\begin{equation}
    C = Z \xrightarrow{H_Z^{\top}} Q \xrightarrow{H_X} X
\end{equation}
where we adopt the convention so that $H_1(C),H^1(C)$ correspond to $Z$-, $X$-type logicals, respectively.
For convenience, denote the weights $\weight_Z,\weight_X,\qubit_Z,\qubit_X$ of code $C$ (corresponding to the max column weight of $\partial,\partial^{\top}$) via the following diagram
\begin{equation}
\label{eq:weight-diagram}
Z \xrightleftharpoons[\qubit_Z]{\weight_Z} Q \xrightleftharpoons[\weight_X]{\qubit_X} X
\end{equation}
If the weights are $O(1)$ for a family of increasing size codes, then we refer to the family as (quantum) \textbf{low-density parity check} or \textbf{(q)LDPC} for short.
\end{example}

\begin{definition}[Systolic Distance]
    Let $C$ be a complex. Then the $i$\textbf{-systolic distance} $d_i(C)$) is 
    \begin{equation}
        d_i(C) = \min_{\ell:0\ne [\ell]\in H_i(C)} |\ell|
    \end{equation}
    Similarly, define the \textbf{(co)systolic distance} $d^{i}(C)$ via $H^{i}(C)$. Note that when relating to CSS codes, $d_1(C),d^1(C)$ is the $Z$-, $X$-type code distance.
\end{definition}
\subsection{Maps}

\begin{definition}[Chain Map]
    Let $A,D$ be complexes with $\partial^A,\partial^D$.
    Then a \textbf{chain map} $g:A\to D$ is a collection of linear maps $g_i:A_i\to D_i$ such that $g\partial^A = \partial^D g$, i.e., the following diagram commutes
    \begin{equation}
    \label{eq:chain-map}
    \begin{tikzpicture}[baseline]
    \matrix(a)[matrix of math nodes, nodes in empty cells, nodes={minimum size=25pt},
    row sep=1.5em, column sep=2em,
    text height=1.25ex, text depth=0.25ex]
    {\cdots & A_i  & A_{i-1} & \cdots \\
     \cdots & D_i  & D_{i-1} & \cdots \\};
    \path[->,font=\scriptsize]
    (a-1-1) edge node[above]{} (a-1-2)
    (a-1-2) edge node[above]{$\partial^{A}$} (a-1-3)
    (a-1-3) edge node[above]{} (a-1-4)
    (a-2-1) edge node[above]{} (a-2-2)
    (a-2-2) edge node[above]{$\partial^{D}$} (a-2-3)
    (a-2-3) edge node[above]{} (a-2-4);
    \path[->,font=\scriptsize]
    (a-1-1) edge (a-2-1)
    (a-1-2) edge node[right]{$g_i$} (a-2-2)
    (a-1-3) edge node[right]{$g_{i-1}$} (a-2-3)
    (a-1-4) edge (a-2-4);
    \end{tikzpicture}
    \end{equation}
\end{definition}

\begin{remark}
    A chain map $g:A\to D$ induces a linear map between homologies $[g]:H_i(A)\to H_i(D)$, where the bracket notation $[g]$ is used so that we have the concise relation $[g][\ell]=[g\ell]$ for any $[\ell]\in H_i(A)$
\end{remark}
\begin{definition}[Cone]
\label{def:cone}
    Let $g:A\to D[-1]$ be a chain map, where $D_{i}[-1]=D_{i-1}$ for all $i$. The \textbf{(mapping) cone} $\cone g$ is the complex with vector spaces $A_i\oplus D_{i}$ and differentials
    \begin{equation}
        \partial =
        \begin{pmatrix}
            \partial^{A} &\\
            g & \partial^{D}
        \end{pmatrix}
    \end{equation}
    We denote this by the following diagram
    \begin{equation}
        \label{eq:cone}
        \begin{tikzpicture}[baseline]
        \matrix(a)[matrix of math nodes, nodes in empty cells, nodes={minimum size=25pt},
        row sep=1.5em, column sep=2em,
        text height=1.25ex, text depth=0.25ex]
        {&\cdots & A_i   & A_{i-1} & \cdots \\
         \cdots & D_i  & D_{i-1} & \cdots &\\};
        \path[->,font=\scriptsize]
        (a-1-2) edge (a-1-3)
        (a-1-3) edge node[above]{$\partial^{C}$} (a-1-4)
        (a-1-4) edge (a-1-5)
        (a-2-1) edge (a-2-2)
        (a-2-2) edge node[above]{$\partial^{D}$} (a-2-3)
        (a-2-3) edge (a-2-4);
        \path[->,font=\scriptsize]
        (a-1-2) edge (a-2-2)
        (a-1-3) edge node[right]{$g$} (a-2-3)
        (a-1-4) edge (a-2-4);
        \end{tikzpicture}
    \end{equation}
\end{definition}

\begin{lemma}[Snake, Lemma 1.5.3 of \cite{Weibel_1994}]
    \label{lem:snake}
    Let $g:A\to D[-1]$ be a chain map and $C=\cone g$ be the cone. Then the following sequence\footnote{A nice mnemonic for applying the Snake Lemma is to remove arrows along rows of Eq.~\eqref{eq:cone}, replace them with the corresponding homologies of the (row) complex, and then add the diagonal arrows as shown in Eq.~\eqref{eq:snake}.} is exact.
    \begin{equation}
    \label{eq:snake}
    \begin{tikzpicture}[baseline]
    \matrix(a)[matrix of math nodes, 
               nodes in empty cells, 
               nodes={anchor=center, inner sep=2pt},
               row sep=2em, 
               column sep={3.2em, between origins},
               text height=1.5ex, text depth=0.25ex]
    {
         \cdots & & H_{i+1}(A) & & H_{i}(A) & & \cdots\\
         & \cdots & & H_i (C) & & \cdots  & \\
        \cdots & & H_{i}(D) & & H_{i-1} (D) & & \\
    };
    \path[->, >=stealth, font=\scriptsize]
        (a-3-1) edge (a-2-2)
        (a-2-2) edge (a-1-3)
        (a-1-3) edge node[right]{$[g]$} (a-3-3)
        (a-3-3) edge node[below right]{$[\iota]$} (a-2-4)
        (a-2-4) edge node[above left]{$[\pi]$} (a-1-5)
        (a-1-5) edge node[right]{$[g]$} (a-3-5)
        (a-3-5) edge (a-2-6)
        (a-2-6) edge (a-1-7);
    \end{tikzpicture}
    \end{equation}
    where $\iota:D\hookrightarrow C$ is the inclusion (chain) map and $\pi:C \to A$ is the projection (chain) map. $[g]$ denotes the map induced by $g$ on homology; similarly for $[\iota]$ and $[\pi]$.
\end{lemma}

\begin{example}[Logical Measurement]
    \label{ex:logical-measurement}
    As shown in Ref. \cite{ide2025fault}, the cone is a mathematical framework used for logical measurement.\footnote{More generally, it was shown in Ref.~\cite{yuan2025unified} that an iterative cone framework can be used to describe any quantum code embedding.} Specifically, if $A,D$ denote the ancilla and data CSS codes, respectively, where $H_1(A)=0$, and $g$ denotes a way to \textit{merge} the two code together, then by the Snake Lemma \ref{lem:snake}, 
    \begin{equation}
        H_1(C) = \coker [g] \equiv  H_1(D)/\im [g]
    \end{equation}
    Moreover, if $g$ is \textit{cleverly chosen} so that $\im [g] = [\ell^{\star}]$ is a specific logical operator, then the logical $\ell^{\star}$ in $D$ is now \textit{contractible} in $C$, i.e., $\ell^{\star}$ is now a stabilizer in $C$ and thus can be measured from syndrome extraction.
\end{example}

\begin{remark}
    While the cone is a useful framework for logical measurement, it is not the only one~\cite{cowtan2024css, huang2022homomorphic}. In particular, there are fault-tolerant logical measurement procedures which are not cones but instead quotients~\cite{poirson2025engineering} or other operations, such as ancilla-free surgery between geometrically enhanced 4D toric codes~\cite{aasen2025geometrically}.
\end{remark}

\begin{example}[Weights]
    \label{ex:weight-reduction}
    Suppose that the data code $D$ is LDPC with $d_1(D)\gg 1$, then to perform logical measurement as in Example \ref{ex:logical-measurement} in a fault tolerant manner, we desire the constructed cone $C$ to also be LDPC with $d_1(C)\gg 1$.
    This suggests that the ancilla $A$ must be LDPC, but more importantly, the map $g$ must be sparse, i.e., bounded in max column and row weight.
    In particular, similar to Eq.~\eqref{eq:weight-diagram}, the weights of $C$ can be obtained using the following 2D diagram
    \begin{equation}
    \label{eq:weight-diagram-cone}
    \begin{tikzpicture}[baseline]
    \matrix(a)[matrix of math nodes, nodes in empty cells, nodes={minimum size=25pt},
    row sep=1.5em, column sep=2em,
    text height=1.25ex, text depth=0.25ex]
    {& A_2  & A_1 & A_0 \\
     D_2  & D_1 & D_0 &\\};
    \path[-left to,font=\scriptsize,transform canvas={yshift=0.2ex}]
    (a-1-2) edge node[above]{}  (a-1-3)
    (a-1-3) edge node[above]{}  (a-1-4)
    (a-2-1) edge node[above]{}  (a-2-2)
    (a-2-2) edge node[above]{}  (a-2-3);
    \path[left to-,font=\scriptsize,transform canvas={yshift=-0.2ex}]
    (a-1-2) edge node[below]{}  (a-1-3)
    (a-1-3) edge node[below]{}  (a-1-4)
    (a-2-1) edge node[below]{}  (a-2-2)
    (a-2-2) edge node[below]{}  (a-2-3);
    \path[-left to,font=\scriptsize,transform canvas={xshift=0.2ex}]
    (a-1-2) edge node[right]{}  (a-2-2)
    (a-1-3) edge node[right]{}  (a-2-3);
    \path[left to-,font=\scriptsize,transform canvas={xshift=-0.2ex}]
    (a-1-2) edge node[left]{}  (a-2-2)
    (a-1-3) edge node[left]{}  (a-2-3);
    \end{tikzpicture}
    \end{equation}
    Meanwhile to show that $d_1(C)\gg 1$, we in general require a relative expansion property\footnote{See, e.g., Definition~2 of~Ref.~\cite{swaroop2026universal} and Cleaning Lemma of Ref.~\cite{yuan2025unified}.} between $A$ and $g$.
    In our case (see Section \ref{sec:attaching-ancilla-cone}), we use a stronger property of $A$ containing an expanding graph.
\end{example}

\begin{remark}
    Note that in Example \ref{ex:logical-measurement}-\ref{ex:weight-reduction}, the  logical $\ell^{\star}$  in general has large weight $\ge d_1(D)$ despite the fact that $C$ is LDPC. Hence, logical measurement can be regarded as a particular application of weight reduction.
\end{remark}

\subsection{Graphs}
\label{sec:graphs}
\begin{definition}[Graphs]
    Given a graph $\sG=(\sV,\sE)$, there is an associated complex $G=E\to V$ where $E,V$ are the $\dF_2$ vector spaces generated by $\sE,\sV$ and the differential map is given by the adjacency relation. 
    Conversely, a complex $G:E\to V$ with basis can define a graph $\sG=(\sE,\sV)$ if $|\partial e|=2$ for all basis elements $e\in E$. In either case, we refer to $G$ as a \textbf{graph} complex with associated graph $\sG$.
    For example, the repetition code $R(L)$ is a graph complex. 
    We further call complex $C=C_2 \to C_1 \to C_0$ a \textbf{cell} complex if $C_1 \to C_0$ is a graph complex.
\end{definition}

\begin{definition}[Union of Graphs]
    Let $\sG^A,\sG^B$ denote graphs with possible overlap in vertices and edges. Then their union $\sG \equiv \sG^A \cup \sG^B$ with vertices $\sV^A\cup \sV^B$ and edges consisting of the union $\sE^A\cup \sE^B$ is also a graph.
    Equivalently, given graph complexes $G^A,G^B$ with associated graphs $\sG^{A},\sG^{B}$ with possible overlap, let the \textbf{union} $G$ be the graph complex with the associated union of graph $\sG^{A} \cup \sG^{B}$. 
\end{definition}

\begin{definition}[Subgraph]
    A graph complex $G$ is a \textbf{subgraph} of a graph complex $G':E'\to V'$ with \textbf{embeddings} $\iota:G\to G'$ if $\iota$ is injective and a chain map.
\end{definition}

\begin{remark}
    If $G$ is the union of graph complexes $G^A,G^B$, then it's clear that $G^A,G^B$ are subgraphs of $G$ with the natural inclusion map as the embedding.
\end{remark}

\begin{definition}[Deformation]
    \label{def:deformation}
    Let $G$ be a graph complex and $L\ge 1$. 
    For every edge $e=v_0 v_1 \in \sE$, let $\gamma(e)$ be a path graph connecting $v_0 \leftrightarrow v_1$ of edge-length $L$.
    Viewing $\gamma (e)$ as graph complexes, a $L$\textbf{-deformation} $\gamma G$ of $G$ is the union of all $\gamma (e),e\in \sE$.
    We refer to $\gamma$ as the \textbf{deformation} map.
    Moreover, if every edge in $\gamma G$ is in at most $\rho$ many paths $\gamma(e),e\in \sE$, we say the the deformation has \textbf{congestion} $\rho$.
    
\end{definition}

\begin{remark}
    If $\gamma(e)$ are simple paths and disjoint (except at endpoints of $e$) for distinct $e\in \sE$, then the deformation is a subdivision of $G$ is the conventional sense, i.e., each edge $e \in \sE$ is subdivided to obtain $\gamma(e)$.
\end{remark}

\section{Parsimonious Cones}
\label{sec:cone}
This section aims to provide an explicit construction of Ref. \cite{hsieh2025simplified}.
Specifically, Sections \ref{sec:graphs}-\ref{sec:direct-embedding-instead-of-subdivision} explicitly constructs the necessary low overhead ancilla, resulting in Theorem \ref{thm:cellulated-cone}.


To help readers understand the technical steps in this section, we briefly summarize our construction at a high level. 
Without loss of generality, suppose we have a bipartite graph $G$ with vertex partition $\sV=\sV_0\sqcup \sV_1$ which we wish to embed into a sparse cone. 
Assume for simplicity that $|\sV_0| = |\sV_1| = 2^h$.
First consider, for the sake of intuition, the graph we obtain by attaching a binary tree of height $h$ to each side of the bipartition. 
In other words, attach to $G$ two binary trees of height $h$, one with leaves $\sV_0$ and one with leaves $\sV_1$.
Next, for every non-leaf vertex in the left tree, add an edge connecting it to the same non-leaf vertex in the right tree. 
Conceptually, we observe that this new graph is similar to the cone in Fig.~\ref{fig:cone-replacement}, where all the cycles of the original graph $G$ can be generated by the ``vertical'' faces/cycles in the new graph. 
In this way, we have coned the bipartite graph $G$ into a cell complex, where all vertices have bounded degree. 

This cell complex has several issues, of course. Notably, the vertical faces may have size $O(\log n)$, and an edge (in the binary trees) may be incident to a large number of such faces. 
In some sense, these issues are due to the fact that the vertices in $\sV_0$ may be connected to vertices in $\sV_1$ arbitrarily (in the original graph $G$), and our complex above is only capable of coning connections ``local'' to the binary trees, but not ``global'' connections.
Specifically, vertices $u\in \sV_0$ and $v\in \sV_1$ may be connected in $G$ but have distance $O(\log n)$ on the binary tree, which induces a large vertical $2$-cell; a large number of these $2$-cells may overlap significantly.
We therefore need to sparsify this cone.

The method from Refs.~\cite{hsieh2025simplified,berdnikov2022parsimonious} effectively  ``stretches'' the above cell complex by a factor of $O(\log n)$ by building a sequence of $O(\log n)$ trees. Specifically, an edge in $G$ will be stretched into a length $O(\log n)$ path. 
The trees are consecutively connected, but not in the trivial 1-to-1 manner as described above. 
Instead, vertices in consecutive trees are connected in an ``interpolating'' manner, so that arbitrary connections in $G$ are instantiated by a sequence of $O(\log n)$ local connections between binary trees. 
Since the connections between binary trees are local, the vertical faces have a bounded-weight generating set, and every edge is only incident to a constant number of the generating $2$-cells. 
This final cone is the output of our construction, which can be used for logical measurement. 

In this section, we first describe how to build a sequence of interpolating, consecutively connected trees such that the face-to-edge incidence matrix is sparse. 
By applying this construction to a bipartite graph $G$ in Sec.~\ref{sec:raw_cones}, we obtain a large cone complex of size $O(|\sV|^2\log |\sV|)$ in which a deformation of $G$ is embedded. 
Most of this cone complex turns out to be extraneous, and in Sec.~\ref{sec:pruning} we prune this complex down to size $O(|\sV|\log |\sV|)$ while preserving its important properties. 




\subsection{Raw Cones}
\label{sec:raw_cones}
Given a binary tree, a standard method to label the vertices is to start from the root via an empty label $|\varnothing\ket$, continue along the branches via labels $|s_1\ket$ and iterate so that vertices at height $\ell$ are labeled via $|s_1\cdots s_\ell\ket$ where $s_i\in \{0,1\}$.
In this case, $s_1\cdots s_\ell s_{\ell+1}\cdots$ are branches of $s_1\cdots s_{\ell}$. 
More generally, however, there's no reason to put the $\ell$ bit at position $\ell$, e.g., one can label vertices at height 2 via $|s_2 s_1\ket$ instead.
In fact, the labeling of adjacent heights do not even need to be consistent, e.g., vertices at height 2 can be labeled via $|s_2 s_1\ket$ while vertices at height 3 can be labeled via $|s_1 s_2 s_3\ket$.
Therefore, to include the most general scenario, we need the following definition.

\begin{definition}
    Let $\tau_\ell$ be a permutation on $\ell$ bits and $s=s_1\cdots s_\ell \in \{0,1\}^\ell$ be a string of bits. Then denote $\tau s$ the string obtained from $s$ via ordering the positions of bits based on the permutation $\tau$, i.e., $\tau s =s_{\bar{\tau}(1)}\cdots s_{\bar{\tau}(\ell)}$ where $\bar{\tau}$ is the inverse permutation of $\tau$ and we have omitted the subscript $\ell$ when the length of the bit string is clear.
    Also write $\| s\|=\ell$ to denote the \textbf{length} of the bit string (different from the Hamming weight).
\end{definition}

\begin{definition}[Binary tree with labels]
    Let $\bm{\tau} = (\tau_{1},...,\tau_{h})$ denote a sequence of permutations $\tau_\ell$ on $1\le \ell\le h$ bits. 
    A \textbf{binary tree} $T$ of height $h$ with \textbf{labels} $\bm{\tau}$ is that with vertices labeled by $|\bar{\tau}s\ket$ where $s\in \{0,1\}^\ell$ are bit strings of length $0\le \ell \le h$ and edges $\| \bar{\tau}s\ket$ where $s$ are bit strings of length $1\le \ell \le h$, and adjacency relation
    \begin{equation}
        \label{eq:tree-adjacency}
        \partial^T \| \bar{\tau}s\ket = |\bar{\tau}s\ket + |\bar{\tau} \omega s\ket,
    \end{equation}
    where $\omega s = s_1\cdots s_{\ell-1}$ removes the last bit.
    If $\tau_\ell$ is the trivial permutation for all $\ell$, then we say that the binary tree has \textbf{standard labels}.
\end{definition}

\begin{remark}
    Note that given an edge $\| s'\ket$ on a binary tree with labels $\bm{\tau}$, we must first determine the \textit{standard string}, i.e., $s' = \bar{\tau}s$, from which $\omega s$ can then be defined so that the adjacency relation $\partial^T$ is well-defined.
    Also note that $|\bar{\tau}(s_1\cdots s_{\ell} s_{\ell+1}\cdots)\ket$ are descendents of $|\bar{\tau} (s_1\cdots s_{\ell})\ket$.
\end{remark}
\begin{definition}[SWAP]
    Let $\{\sigma_{i,i+1},i\in I\}$ denote a collection of permutations $\sigma_{i,i+1}$ which swap bits $i\lr i+1$ and $I$ is a collection of indices $i$ such that $\{i,i+1\},\{j,j+1\}$ are disjoint for distinct $i,j\in I$. 
    Then the collection induces a sequence $\bm{\sigma}=(\sigma_1,...,\sigma_h)$ of permutations such that
    \begin{equation}
        \sigma_\ell = \prod_{i\in I:i < \ell} \sigma_{i,i+1}.
    \end{equation}
    We refer to $\bm{\sigma}$ as a \textbf{(non-interacting) SWAP} with indices $I$. Given a sequence of permutations $\tau$, we then abuse notation and write $\bm{\sigma\tau}$ as the sequence of permutations $(\sigma\tau)_\ell,1\le \ell \le h$ obtained via entry-wise multiplication/composition, i.e., $(\sigma\tau)_\ell =\sigma_\ell \tau_\ell$.
\end{definition}

The following definition and the subsequent Lemma~\ref{lem:interpolation-cone} are the most important primitives used in our construction. 
Given two binary trees $T',T$ with labels $\bm{\sigma\tau}, \bm{\tau}$, we create a tree $S$ where the vertices at the heights $i\in I$, i.e. the indices of the non-interacting SWAP $\bm{\sigma}$, are replaced by non-branching points (see Fig. \ref{fig:interpolation-cone}). 
We then locally connect the trees $T'$ with $S$, and $S$ with $T$. 
Returning to the intuition described at the start of this section, these local connections generate slightly non-local connections between the leaves of $T'$ and $T$, while ensuring the $2$-cells to $1$-cells incidence is sparse. 
By interpolating with $O(\log n)$ such trees, we will be able to generate arbitrary connections between the leaves of the first and last trees through these local connections.

\begin{definition}[Interpolating via SWAP]
    Let $T',T$ be binary trees with labels $\bm{\sigma\tau}, \bm{\tau}$, respectively, where $\bm{\sigma}$ is a non-interacting SWAP with indices $I$. 
    Let $s\in\{0,1\}^\ell$ and $\sigma^\star$ denote the map
    \begin{equation}
        \sigma^{\star} s = 
        \begin{cases}
            s \quad &\text{if }\ell\notin I, \\
            \omega(s)\star \quad &\text{if }\ell\in I.
        \end{cases}
    \end{equation}
    Here $\star$ is just an arbitrary placeholder for index $\ell$ so that $\omega(s)\star$ is a string of length $\ell$.
    In particular, the $\star$ notation denotes a non-branching point, as shown in Fig. \ref{fig:interpolation-cone}.
    Note that $\omega \sigma^{\star} = \omega$.
    Define the \textbf{interpolation} complex $S$ between $T'$ and $T$ as the graph complex with vertices $|\bar{\tau}\sigma^{\star} s\ket$ for $s\in \{0,1\}^\ell,0\le \ell \le h$ and edges $\| \bar{\tau} \sigma^{\star} s\ket$ for $s\in \{0,1\}^\ell,1\le \ell \le h$ and adjacency relation
    \begin{equation}
        \label{eq:interpolation-adjacency}
        \partial^{S} \| \bar{\tau} \sigma^{\star} s\ket = |\bar{\tau} \sigma^{\star} s\ket +|\bar{\tau}\sigma^{\star} \omega s\ket.
    \end{equation}
\end{definition}


\begin{figure}[ht]
\centering
\subfloat[]{%
    \centering
    \includegraphics[width=0.7\columnwidth]{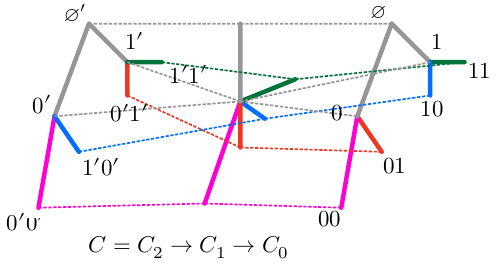}
}
\subfloat[\label{fig:input-graph}]{%
    \centering
    \includegraphics[width=0.22\columnwidth]{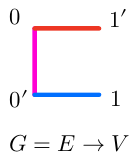}
}
\caption{Interpolation Cone. (a) Left $T'$ and right $T$ binary trees denotes those with labels $\bm{\tau}'=\bm{\sigma\tau}, \bm{\tau}$, respectively, where $\bm{\sigma}$ is a SWAP, while the middle denotes the interpolation graph complex $S$ of $T',T$. The vertices of $T',T$ are labeled by $|s'\ket,|s\ket$ with bit strings $s$, respectively.
Note that, e.g., $|1'0'\ket = |\tau'(01)\ket$.
(b) denotes the graph $G$ from which $C$ is constructed as described in Theorem~\ref{thm:parsimonious-cone}.
}
\label{fig:interpolation-cone}
\end{figure}

\begin{remark}[Crucial Relation]
    \label{rem:crucial-relation}
    Note the crucial algebraic relation $\sigma^{\star} \omega \sigma=\sigma^{\star} \sigma\omega$. In fact, $\sigma^{\star} \omega^{m} \sigma = \sigma^{\star} \sigma \omega^{m}$ where $\omega^{m} = \omega \omega \cdots \omega$ is the composition of $m$ truncations $\omega$.
\end{remark}

The following theorem is illustrated by the example in Fig. \ref{fig:interpolation-cone}.
To describe the $2$-cells (faces) and $1$-cells (edges) between consecutive trees, we define intermediate $1$-complexes and chain maps such that the mapping cone (Def.~\ref{def:cone}) is the desired $2$-complex.

\begin{lemma}[Interpolation Cone, Fig. \ref{fig:interpolation-cone}]
    \label{lem:interpolation-cone}
    Let $T',T$ be binary trees with labels $\bm{\sigma\tau}, \bm{\tau}$, respectively, where $\bm{\sigma}$ is a non-interacting SWAP with indices $I$, and $S$ denote the interpolation complex.
    Let $T'S,ST$ be binary trees with labels $\bm{\sigma\tau},\bm{\tau}$, respectively, i.e., copies of $T',T$, respectively.
    Define $g_i:(T'S \oplus ST)_i \to (T'\oplus S\oplus T)_i$ for $i=1,0$ via
    \begin{align}
        g_1\| \bar{\tau} \sigma s, T'S\ket &= \| \bar{\tau} \sigma s, T'\ket + \| \bar{\tau} \sigma^{\star} \sigma s, S\ket \\
        g_1\| \bar{\tau} s, ST\ket &= \| \bar{\tau} \sigma^{\star} s, S\ket + \| \bar{\tau} s, T\ket
    \end{align}
    and similarly\footnote{Where $\|\cdots \ket$ is replaced with $|\cdots\ket$.} for $g_0$.
    Then $g$ is a chain map and the cone complex $C=\cone (g)$ -- see Fig. \ref{fig:interpolation-cone} -- satisfies
    \begin{align}
        \label{eq:cone-props1}
        H_2(C)=H_1(C)=0, &\quad H_0(C) \cong \dF_2,\\
        \label{eq:cone-props2}
        C_2 \xrightleftharpoons[4]{4} C_1 \xrightleftharpoons[9]{2} C_0.
    \end{align}
\end{lemma}
\begin{proof}
    Let us consider the map $g$ via the following diagram
    \begin{equation}
    \label{eq:interpolation-diagram}
    \begin{tikzpicture}[baseline]
    \matrix(a)[matrix of math nodes, nodes in empty cells, nodes={minimum size=25pt},
    row sep=2em, column sep=4em,
    text height=1.25ex, text depth=0.25ex]
    {(T'S\oplus ST)_1  & (T'S\oplus ST)_0 \\
    (T'\oplus S\oplus T)_1  & (T'\oplus S\oplus T)_0 \\};
    \path[->,font=\scriptsize]
    (a-1-1) edge node[above]{$\partial^{T'S} \oplus \partial^{ST}$} (a-1-2)
    (a-2-1) edge node[above]{$\partial^{T'}\oplus \partial^{S} \oplus \partial^{T}$} (a-2-2);
    \path[->,font=\scriptsize]
    (a-1-1) edge node[right]{$g_1$} (a-2-1)
    (a-1-2) edge node[right]{$g_0$} (a-2-2);
    \end{tikzpicture}
    \end{equation}
    Note the crucial relation $\sigma^{\star} \omega \sigma=\sigma^{\star} \sigma\omega$ and thus $g$ is a chain map. Specifically,
    \begin{align}
        g_0 \partial^{T'S} \| \bar{\tau} \sigma s,T'S\ket &= |\bar{\tau}\sigma s, T'\ket+|\bar{\tau}\sigma\omega s, T'\ket \\
        &\quad +|\bar{\tau} \sigma^{\star} \sigma s, S\ket +|\bar{\tau} \sigma^{\star} \sigma \omega s, S\ket \nonumber\\
        &= \partial^{T'\oplus S} g_1 \| \bar{\tau} \sigma s, T'S\ket,
    \end{align}
    and the case is similar for edges in $ST$.
    Since $T'$ is a tree, $H_1(T')=0$ (no cycles) and $H_0(T') \cong \dF_2$ (connected) where any single vertex of $T'$ is a representation of $H_0(T')$. 
    The case is similar for $S,T,T'S,ST$. 
    Hence, by the cone structure, the statement follows.
    Specifically, the Snake Lemma implies that the following is an exact sequence, where $A=T'S\oplus ST$ and $B=T'\oplus S\oplus T$ for shorthand (compare structure with Eq.~\eqref{eq:interpolation-diagram}).
    \begin{equation}
    \label{eq:interpolation-diagram-snake}
    \begin{tikzpicture}[baseline]
    \matrix(a)[matrix of math nodes, 
               nodes in empty cells, 
               nodes={anchor=center, inner sep=2pt},
               row sep=2em, 
               column sep={3.2em, between origins},
               text height=1.5ex, text depth=0.25ex]
    {
         & & H_1(A)=0 & & H_0(A) & & 0\\
         & H_2(C) & & H_1 (C) & & H_0(C) & \\
        0 & & H_1(B)=0 & & H_0 (B) & & \\
    };
    \path[->, >=stealth, font=\scriptsize]
        (a-3-1) edge (a-2-2)
        (a-2-2) edge (a-1-3)
        (a-1-3) edge (a-3-3)
        (a-3-3) edge (a-2-4)
        (a-2-4) edge (a-1-5)
        (a-1-5) edge node[right]{$[g_0]$} (a-3-5)
        (a-3-5) edge (a-2-6)
        (a-2-6) edge (a-1-7);
    \end{tikzpicture}
    \end{equation}
    Thus $H_1(C) \cong \ker [g_0]$ and $H_0(C) \cong \coker [g_0]$.
    Note that $[g_0]$ is equal to the differential of the repetition code (edges $\to$ vertices) with two bits and three checks, and thus the statement follows.

    Finally, note that the weight of the cone complex is obtained via the following diagram, 
    \begin{equation}
    \label{eq:weight-diagram-interpolation-cone}
    \begin{tikzpicture}[baseline]
    \matrix(a)[matrix of math nodes, nodes in empty cells, nodes={minimum size=25pt},
    row sep=2em, column sep=2em,
    text height=1.25ex, text depth=0.25ex]
    {(T'S\oplus ST)_1 & (T'S\oplus ST)_0 \\
    (T'\oplus S\oplus T)_1 & (T'\oplus S\oplus T)_0\\};
    \path[-left to,font=\scriptsize,transform canvas={yshift=0.2ex}]
    (a-1-1) edge node[above]{$2$}  (a-1-2)
    (a-2-1) edge node[above]{$2$}  (a-2-2);
    \path[left to-,font=\scriptsize,transform canvas={yshift=-0.2ex}]
    (a-1-1) edge node[below]{$3$}  (a-1-2)
    (a-2-1) edge node[below]{$5$}  (a-2-2);
    \path[-left to,font=\scriptsize,transform canvas={xshift=0.2ex}]
    (a-1-1) edge node[right]{$2$}  (a-2-1)
    (a-1-2) edge node[right]{$2$}  (a-2-2);
    \path[left to-,font=\scriptsize,transform canvas={xshift=-0.2ex}]
    (a-1-1) edge node[left]{$4$}  (a-2-1)
    (a-1-2) edge node[left]{$4$}  (a-2-2);
    \end{tikzpicture}
    \end{equation}
    where the number denotes the maximum column weight of the corresponding map.
\end{proof}

\begin{remark}[Interpolation]
    \label{rem:interpolation}
    Note that $C$ is a cell complex.
    Also note that for leaves $s'\in \{0,1\}^h$ so that $\sigma^{\star} s'=s'$,
    the complex $C$ can be regarded as the interpolation which connects vertices $|s',T'\ket$ and $|s',T\ket$ via unique edges passing through $|s',S\ket$ so that for distinct $s'$, the connecting paths are vertex-disjoint.
    A depiction is shown in Fig. \ref{fig:interpolation-cone}.
\end{remark}

\begin{remark}[Interpolating Sequence]
    \label{rem:interpolation-sequence}
    Let $T^0,...,T^n$ denote a sequence of binary trees with labels $\bm{\tau}^0,...,\bm{\tau}^n$ such that adjacent labels $\bm{\tau}^{i},\bm{\tau}^{i-1}$ differ by a non-interacting SWAP, i.e., $\bm{\tau}^{i} = \bm{\sigma}^{i} \bm{\tau}^{i-1}$ for a SWAP $\bm{\sigma}^{i}$. 
    We refer to $T^0,...,T^n$ as a \textbf{SWAP sequence}. 
    Then a similar construction can be obtained via a chain map $g: T^n S^{n} \oplus S^{n} T^{n-1} \cdots \oplus S^1 T^0 \to T^n \oplus \cdots \oplus T^0$ so that $C=\cone(g)$ is a cell complex which satisfies Eqs.~\eqref{eq:cone-props1}-\eqref{eq:cone-props2}.
    We refer to $C$ as the cone complex \textbf{interpolating} the SWAP sequence $T^0,...,T^n$. Similarly, vertices $|s,T^i\ket,i=0,...,n$ are connected together via the interpolation as described previously.
\end{remark}


As depicted in Fig. \ref{fig:interpolation-cone}, given a bipartite graph $G$, a parsimonious cone $C$ is constructed so that a deformation of $G$ is embedded within $C$. 
The formal statement is given as follows.

\begin{theorem}[Parsimonious Cone]
    \label{thm:parsimonious-cone}
    Let $G=E\to V$ denote the complex of a bipartite graph with vertex partitions $\sV=\sV_0\sqcup \sV_1$ where $|\sV_i| \le 2^{h_i}$ and max degree $\Delta(G)$.
    Let $h=h_0+h_1$ and $\tau_h$ denote the permutation on $h$ bits which maps 
    \begin{equation}
        (s_1\cdots s_{h_0})(s_{h_0+1}\cdots s_h)\mapsto (s_{h_0+1}\cdots s_{h}) (s_1 \cdots s_{h_0})
    \end{equation}
    Then there exists a sequence of non-interacting SWAPs $\bm{\sigma}^{1},...,\bm{\sigma}^{h}$ such that if $T^{i}$ are binary trees with labels $\bm{\tau}^{i}\equiv \bm{\sigma}^{i}\cdots \bm{\sigma}^{1}$ ($T^{0}$ is that with trivial labels), then $(\bm{\tau}^{h})_h = \tau_h$.
    Moreover, if $C$ is the cone complex interpolating the sequence $T^0,...,T^{h}$, then there exists a $3h$-deformation of $G$ with congestion $\Delta(G)$ which is a subgraph of $C_1\to C_0$ with 
    \begin{equation}
        \dim C_0 = O(h 2^h) = O(|\sV|^2 \log |\sV|).
    \end{equation}
\end{theorem}
\begin{remark}
    Technically, the permutation can be obtained in $h-1$ SWAPs, i.e., $\bm{\sigma}^{2},...,\bm{\sigma}^{h}$, so that only a $(3h-2)$-deformation of $G$ is required (see Fig. \ref{fig:interpolation-cone}). 
    However, to remove dangling constants, we have added an arbitrary trivial SWAP, i.e., $\bm{\sigma}^{1}$ has empty indices.
\end{remark}

\begin{proof}
    Without loss of generality, assume that $h_0\le h_1$.
    Note that the vertices $\sV_0$ can be labeled via a string bit of length $h_0$, i.e., $\hat{s}$ defines an injective mapping $\sV_0\to \{0,1\}^{h_0}$ for $\sV_0$ and similarly for $\sV_1$. 
    Since the graph is bipartite, every edge $e\in \sE$ with adjacent vertices $v_0\in \sV_0,v_1\in \sV_1$ can be labeled via the string bit $\hat{s}(e)\equiv\hat{s}(v_0)\hat{s}(v_1)$ of length $h$.
    Note that we are overloading the notation $\hat{s}$.
    Consider a binary tree $T^1$ of height $h$ with standard labels so that the edges $e$ can be embedded into the leaves at height $h$ via the map $e\mapsto \| \hat{s}(e),T^1\ket$. 

    Note that the edge $e$ is not concerned with the order of $v_0,v_1$ and thus every edge can also be labeled via $\hat{s}(v_1)\hat{s}(v_0)=\bar{\tau}_h \hat{s}(e) $.
    Note that $\tau_h$ can be induced via $h$ SWAPs in the following sense. Let $\bm{\sigma}^{1}$ denote the trivial SWAP, i.e., that with empty indices, and $\bm{\sigma}^{i},i\ge 2$ denote the SWAP with indices 
    \begin{itemize}
        \item $\{h_0-i+2,h_0-i+4,...,h_0+i-2\}$ for $2\le i \le h_0+1$,
        \item $\{i-h_0,i-h_0+4,...,\min(i+h_0-2,h-1)\}$ for $h_0+1 < i\le h$.
    \end{itemize}
    Further let $\bm{\tau}^{i} = \bm{\sigma}^{i} \cdots \bm{\sigma}^{1}$ so that $(\bm{\tau}^{h})_h=\tau_h$.
    Let $T^{i}$ be the binary trees with labels $\bm{\tau}^i$ and thus the edges $e$ can also be embedded into the leaves at height $h$ via the map $e\mapsto \|\bar{\tau}^{h}\hat{s}(e),T^{h}\ket$ where we note that $\bar{\tau}^h \hat{s} (e) = \bar{\tau}_h \hat{s}(e)$.
    
    As described via Remark \ref{rem:interpolation} and \ref{rem:interpolation-sequence}, we now connect vertices of $T^0$ and $T^{h}$ through interpolation.
    Specifically, let $C$ denote the cone complex interpolating the sequence so that $|s,T^0\ket,|s,T^{h}\ket$ are connected via a path of edge-length $2h$. 
    In particular, $|\hat{s}(v_0)\hat{s}(v_1),T^0\ket$ and $|\hat{s}(v_0)\hat{s}(v_1),T^h\ket = |\bar{\tau}^h(\hat{s}(v_1)\hat{s}(v_0)),T^h\ket$ are connected via a path $\gamma^\mathrm{inter}(e)$ of edge-length $2h$ which is disjoint from $T^0,T^h$ except at the endpoints. 
    Specifically, $\gamma^\mathrm{inter}(e)$ can also be regarded as a sequence of adjacent vertices in the cell complex $C$, i.e.,
    \begin{align}
        \gamma^\mathrm{inter}(e) &= |\hat{s}(e),T^0\ket\to|\hat{s}(e),S^{1}\ket \to |\hat{s}(e),T^{1}\ket \\
        &\quad \to\cdots\to|\hat{s}(e),T^{h}\ket.
    \end{align}

    Also note that $|\hat{s}(v_0)\hat{s}(v_1),T^0\ket$ is a leaf of $|\hat{s}(v_0),T^0\ket$ and thus there exists a connected path $\gamma^0(e)$ contained entirely in $T^0$ of edge-length $h_1$ connecting the two vertices, i.e.,
    \begin{align}
        \gamma^0(e) &= |\hat{s}(v_0),T^0\ket \to |\hat{s}(v_0)\hat{s}(v_1)_1,T^0\ket \\
        &\;\to |\hat{s}(v_0)\hat{s}(v_1)_1\hat{s}(v_1)_2,T^0\ket \cdots  \to |\hat{s}(e),T^0\ket. \nonumber
    \end{align}
    Similarly, there exists a connected path $\gamma^{h}(e)$ contained entirely in $T^h$ of edge-length $h_0$ connecting $|\bar{\tau}^h(\hat{s}(v_1)\hat{s}(v_0)),T^h\ket$ to $|\bar{\tau}^h\hat{s}(v_1),T^h\ket$.
    Let $\gamma(e)$ denote the concatenation of paths $\gamma^0(e)\to \gamma^\mathrm{inter}(e)\to\gamma^h(e)$ so that $\gamma(e)$ has edge-length $3h$.
    Note that $\gamma^{\mathrm{inter}}(e)$ are simple paths and vertex-disjoint for distinct edges $e$, and thus one may naively expect that $\gamma$ induces a subdivision of $G$.
    Unfortunately, this is not true.
    Rather, it's possible for $\gamma^{0}(e),\gamma^{h}(e)$ to overlap for distinct $e$ in $G$ and thus $\gamma$ only induces a $3h$-deformation of $G$.
    Specifically, we identify $\sV_0\ni v_0 \mapsto |\hat{s}(v_0),T^1\ket$ and $\sV_1\ni v_1\mapsto |\bar{\tau}^{h}\hat{s}(v_1),T^{h}\ket$ so that $\gamma G =\bigcup_{e}\gamma(e)$ is a $3h$-deformation of $G$. 
    The inclusion map $\gamma G \hookrightarrow C$ shows that $\gamma G$ is a subgraph of $C$.

    As remarked, $\gamma^{0}(e),\gamma^{h}(e)$ may overlap for distinct $e$ in $G$. However, the overlap is clearly bounded by $\Delta (G)$ and thus $\gamma G$ has congestion $\Delta(G)$.
\end{proof}


\begin{corollary}
    \label{cor:parsimonious-cone}
    Let $G$ denote a graph complex with max degree $\Delta(G)$ and let vertices $|\sV|\le 2^{h_0}$ and edges $|\sE| \le 2^{h_1}$ and $h=h_0 +h_1$.
    Then there exists cell complex $C$ such that a $6h$-deformation of $G$ with congestion $\Delta(G)$ is a subgraph of $C$ with 
    \begin{equation}
        \dim C_0 =O(h2^h)= O(|\sV| |\sE| \log(|\sV||\sE|)).
    \end{equation}
    In particular, if $\Delta(G)=O(1)$, then 
    \begin{equation}
        \dim C_0= O(|\sV|^2 \log |\sV|).
    \end{equation}
\end{corollary}
\begin{proof}
    Note that if we subdivide each edge $e$ in $G$ by adding a vertex in the center, then the subdivision is a bipartite graph with vertex partition corresponding to $\sV\sqcup \sE$, i.e., the 2-subdivision of $G$ is bipartite.
    The statement then follows.
\end{proof}

\subsection{Pruning for Better Cones}

\label{sec:pruning}
In Theorem \ref{thm:parsimonious-cone} (and Corollary \ref{cor:parsimonious-cone}), the cone complex has size $O(|\sV|^2 \log|\sV|)$ -- it is not parsimonious enough.
In this subsection, we obtain a better bound -- $O(|\sV| \log |\sV|)$ -- via 
\begin{enumerate}[label=\arabic*)]
    \item pruning the unnecessary branches in $C$, which gets the scaling down to $O(|\sV| \log^2 |\sV|)$,
    \item removing/contracting vertices of degree 2, i.e., vertices which do not branch, and corresponding edges, to reach the scaling $O(|\sV| \log |\sV|)$.
\end{enumerate}

Although doing both simultaneously is possible, it makes understanding the proof more difficult. 
Hence, the first subsection provides the proof with only the first step, while the second considers both simultaneously. 
\subsubsection{Pruning extraneous branches}


The intuition behind pruning is simple: our goal is to embed a bipartite graph with $O(|\sV|)$ vertices on each side, while our binary trees are of height $h$ and each have $O(|\sV|^2)$ leaves. 
We therefore remove the branches in the trees which end at irrelevant leaves.

\begin{definition}[Pruned Trees]
    Let $\sL$ denote a collection of \textbf{leaves} at height $h$, i.e., $h$-length string bits in $\{0,1\}^h$.
    Define the $\sL$-\textbf{pruned} binary tree $T$ of height $h$ with labels $\bm{\tau}$ as the graph complex with edges $\| \bar{\tau} s\ket$ and vertices $|\bar{\tau}s\ket$ over $(\ell\le h)$-length bits $s=s_1 \cdots s_\ell$, such that there exists $s_{\ell+1},...,s_{h}\in \{0,1\}$ so that the padded bit string $s_{\mathrm{leaf}}=s_1\cdots s_h$ satisfies $s_{\mathrm{leaf}}\in \tau \sL$, i.e., write $s\subseteq s_{\mathrm{leaf}}$ and refer to $s$ as a \textbf{$\tau \sL$-truncation}.
    In particular, note that if edge $\|\bar{\tau} s\ket$ is well-defined in ${T}$, then so are the vertices $|\bar{\tau} s\ket,|\bar{\tau}\omega s\ket$, and thus the usual adjacency relation in Eq.~\eqref{eq:tree-adjacency} is well-defined.
\end{definition}

\begin{definition}[Pruned Interpolations]
    Let $\sL$ be leaves at height $h$, and $T',T$ be $\sL$-pruned binary trees of height $h$ with labels $\bm{\sigma\tau}, \bm{\tau}$, respectively, where $\bm{\sigma}$ is a SWAP with indices $I$. 
    Define the $\sL$-\textbf{pruned} interpolation complex $S$ between $T',T$ as the graph complex with vertices $|\bar{\tau}\sigma^{\star} s\ket$ and edges $\| \bar{\tau} \sigma^{\star} s\ket$ over all $\tau \sL$-truncations $s$.
    Note that if the edge $\|\bar{\tau} \sigma^\star s\ket$ is well-defined, then so are the vertices $|\bar{\tau} \sigma^{\star} s\ket,|\bar{\tau}\sigma^{\star} \omega s\ket$ and thus the usual adjacency relation as in Eq.~\eqref{eq:interpolation-adjacency} is well-defined.
\end{definition}

\begin{remark}[Pruned Size]
    \label{rem:pruned-size}
    Note that the number of vertices in the pruned tree $T$ and pruned interpolation complex is $=O(h|\sL|)$ in general.
\end{remark}

First note the following straightforward, yet important, observation. 

\begin{lemma}[SWAP and Pruned]
    \label{lem:SWAP-and-prune}
    Let $\sL$ be leaves and $\bm{\tau}$ be a label and $\bm{\sigma}$ be a SWAP with indices $I$. Let $s$ be a $\ell$-bit string. If $\ell\notin I$, then $s$ is a $\tau \sL$-truncation $\Leftrightarrow \sigma s$ is a $\sigma \tau \sL$-truncation -- specifically, if $s\subseteq s_{\mathrm{leaf}}\in \tau\sL$, then $\sigma s \subseteq \sigma s_{\mathrm{leaf}}$.
\end{lemma}

\begin{figure}[ht]
\centering
\subfloat[]{%
    \centering
    \includegraphics[width=0.7\columnwidth]{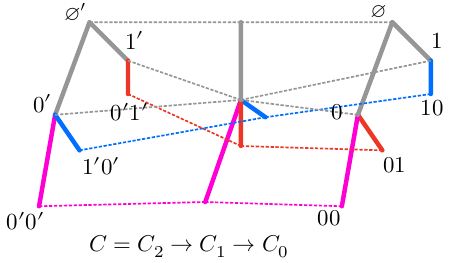}
}
\subfloat[]{%
    \centering
    \includegraphics[width=0.22\columnwidth]{input-graph.pdf}
}
\caption{Pruned Interpolation Cone. In comparison to Fig. \ref{fig:interpolation-cone}, the green branches are not necessary to embed the (subdivision of) graph $G$ within $C$ and are thus pruned.}
\label{fig:pruned-interpolation-cone}
\end{figure}

\begin{lemma}[Pruned Interpolation Cone]
    \label{lem:pruned-interpolation-cone}
    Let $\sL$ be leaves at height $h$.
    Let $T',T$ be $\sL$-pruned trees with labels $\bm{\sigma\tau}, \bm{\tau}$, respectively, where $\bm{\sigma}$ is a SWAP, and let $S$ denote the $\sL$-pruned interpolation with respect to $T',T$.
    Let $T'S,ST$ be $\sL$-pruned binary trees with labels $\bm{\sigma\tau},\bm{\tau}$, respectively.
    Define $g_i:(T'S \oplus ST)_i \to (T'\oplus S\oplus T)_i$ for $i=1,0$ via
    \begin{align}
        g_1\| \bar{\tau} \sigma s, T'S\ket &= \| \bar{\tau} \sigma s, T'\ket + \| \bar{\tau} \sigma^{\star} \sigma s, S\ket, \\
        g_1\| \bar{\tau} s, ST\ket &= \| \bar{\tau} \sigma^{\star} s, S\ket + \| \bar{\tau} s, T\ket,
    \end{align}
    and similarly\footnote{Replace $\|\cdots \ket$ with $|\cdots\ket$.} for $g_0$.
    Then $g$ is a chain map and the cone complex $C=\cone (g)$ -- see Fig. \ref{fig:pruned-interpolation-cone} -- satisfies
    \begin{align}
        \label{eq:pruned-cone-prop1}
        H_2(C)=H_1(C)=0, &\quad H_0(C) \cong \dF_2\\
        \label{eq:pruned-cone-prop2}
        C_2 \xrightleftharpoons[4]{4} C_1 \xrightleftharpoons[9]{2} C_0.
    \end{align}
\end{lemma}

\begin{proof}
    The proof is nearly identical to that of Lemma \ref{lem:interpolation-cone}, and thus we point to the differences.
    In fact, since the algebraic equations are exactly the same, it is sufficient to show that $g_1$ (and similarly, $g_0$) is well-defined, e.g., if $\|\bar{\tau}\sigma s, T'S\ket$ is a valid edge,  then so are $\|\bar{\tau}\sigma s, T'\ket$ and $\|\bar{\tau}\sigma^\star \sigma s, S\ket$.
    Indeed, it's clear that $\|\bar{\tau}\sigma s, T'\ket$ is well-defined and thus we focus on the other edge. 
    Since $T'S$ has labels $\bm{\sigma \tau}$, we see that $s$ is a $\sigma \tau \sL$-truncation.
    Let us consider the following two separate cases, and recall the crucial algebraic relation $\sigma^{\star} \omega \sigma =\sigma^{\star} \sigma \omega$ in Remark \ref{rem:crucial-relation}. Let $s$ be a $\ell$-bit string.
    \begin{itemize}
        \item If $\ell \notin I$, then by Lemma \ref{lem:SWAP-and-prune}, $\sigma s$ is a $\tau \sL$-truncation and thus $\|\bar{\tau} \sigma^\star \sigma s,S\ket$ is well-defined
        \item If $\ell \in I$, then $\ell-1\notin I$ so that $\sigma \omega s=\omega \sigma s$ is a $\tau \sL$-truncation.
        In particular, there exists some bit $b\in\{0,1\}$ such that $\tilde{s} \equiv(\omega \sigma s) b$ is a $\tau \sL$-truncation.
        Note that $\sigma^{\star} \tilde{s} = \sigma^{\star} \sigma s$ and thus $\|\bar{\tau} \sigma^{\star} \sigma s,S\ket$ is well-defined.
    \end{itemize}
    The other cases are similar and thus $g_1,g_0$ are well-defined. 
    The algebraic equations are the same and thus $g$ is a chain map. The statement then follows as in the proof of Theorem \ref{lem:interpolation-cone}.
\end{proof}

\begin{remark}[Pruned Interpolation Sequence]
    \label{rem:pruned-interpolation-seq}
    Similar to Remark \ref{rem:interpolation} and \ref{rem:interpolation-sequence}, $C$ is a cell-complex which connects leaves ($h$-bit strings) $|s',T'\ket$ and $|s',T\ket$ via unique edges passing through $|s',S\ket$ so that for distinct $s'$, the connecting paths are vertex-disjoint.
    A depiction is shown in Fig. \ref{fig:pruned-interpolation-cone}.
    Moreover, if $T^0,...,T^n$ is a sequence of $\sL$-pruned SWAP sequence, then a similar construction can be obtained via a chain map $g: T^n S^{n} \oplus S^{n} T^{n-1} \cdots \oplus S^1 T^0 \to T^n \oplus \cdots \oplus T^0$ so that $C=\cone(g)$ is a cell complex which satisfies Eq.~\eqref{eq:pruned-cone-prop1}-\eqref{eq:pruned-cone-prop2}.
    We refer to $C$ as the $\sL$-pruned cone complex \textbf{interpolating} the $\sL$-pruned SWAP sequence $T^0,...,T^n$. Similarly, vertices $|s',T^i\ket,i=0,...,n$ are connected together via the interpolation as described previously.
\end{remark}

Using the pruned interpolation cone, we can similarly derive Theorem \ref{thm:parsimonious-cone} and its corollary, but with better overhead, which follows from Remark \ref{rem:pruned-size}.
See Fig. \ref{fig:pruned-interpolation-cone} for reference.

\begin{theorem}[Pruned Parsimonious Cone, Fig. \ref{fig:pruned-interpolation-cone}]
    \label{thm:pruned-parsimonious-cone}
    Let $G$ denote the complex of a bipartite graph with vertex partitions $\sV=\sV_0\sqcup \sV_1$ where $|\sV_i| \le 2^{h_i}$, and $G$ has no isolated vertices with max degree $\Delta(G)$
    Then there exists cell complex $C$ such that there exists a $3h$-deformation of $G$ with congestion $\Delta(G)$ which is a subgraph of $C$ where
    \begin{equation}
        \dim C_0 = O(|\sE| \log^2(|\sV||\sE|)),
    \end{equation}
    and 
    \begin{align}
        H_2(C)=H_1(C)=0, &\quad H_0(C) \cong \dF_2\\
        C_2 \xrightleftharpoons[4]{4} C_1 \xrightleftharpoons[9]{2} C_0.
    \end{align}
    In particular, if $\Delta(G)=O(1)$ then 
    \begin{equation}
        \dim C_0 = O(|\sV| \log^2|\sV|).
    \end{equation}
\end{theorem}

\begin{proof}
    The proof is nearly identical to that of Theorem \ref{thm:parsimonious-cone}, except for Remark \ref{rem:pruned-size}, and thus we sketch the idea.
    Arbitrarily label $\sV_i$ via $\hat{s}:\sV_0\to \{0,1\}^{h_i}$ so that edges $e=v_0 v_1\in \sE$ are labelled via $h$-bit strings $\hat{s}(e) =\hat{s}(v_0)\hat{s}(v_1)$.
    Let $\sL=\hat{s}(\sE)$ be the collection of leaves. 
    Then, similar to before, there exists a sequence of SWAPs $\bm{\sigma}^{1},...,\bm{\sigma}^{h}$ such that $T^{i}$ are $\sL$-pruned binary trees with labels $\bm{\tau}^{i}\equiv \bm{\sigma}^{i}\cdots \bm{\sigma}^{1}$ and $(\tau^{h})_h = \tau_h$.
    Let $C$ be the $\sL$-pruned cone complex interpolating $T^0,...,T^{h}$.
    Since $G$ does not have isolated vertices, all vertices will be contained in the $\sL$-pruned cone complex, and thus the statement follows.
\end{proof}

\begin{corollary}
    \label{cor:pruned-parsimonious-cone}
    Let $G$ denote a graph complex of bounded degree $\Delta(G)=O(1)$.
    Then there exists cell complex $C$ such that there exists a $6h$-deformation of $G$ with $O(1)$ congestion which is a subgraph of $C$ with 
    \begin{equation}
        \dim C_0 = O(|\sV| \log^2|\sV|),
    \end{equation}
    and 
    \begin{align}
        H_2(C)=H_1(C)=0, &\quad H_0(C) \cong \dF_2\\
        C_2 \xrightleftharpoons[4]{4} C_1 \xrightleftharpoons[9]{2} C_0.
    \end{align}
\end{corollary}
\begin{proof}
    If $G$ has isolated vertices, then we can add additional edges and connect all isolated vertices into a chain so that $G$ is a subgraph of the new graph complex $G'$. It's clear that $G'$ is bounded in degree, has no isolated vertices, and thus the statement follows.
\end{proof}
\subsubsection{Pruning and Contracting}

On top of pruning, we further contract vertices of degree $2$ in the binary trees to reduce the overhead by another factor of $O(\log |\sV|)$.

\begin{definition}[Pruned* Trees]
    Let $\sL$ be leaves at height $h$ and $\bm{\tau}$ be labels.
    Consider $0<\ell <h$ length bit string $s$.
    If both $s0$ and $s1$ are also $\tau \sL$-truncations, then we say $s$ is $\tau \sL$-\textbf{branching}\footnote{Note that $\tau \sL$-branching implies $\tau \sL$-truncation.}.
    To keep notation simple, we arbitrarily say that the empty bit $\varnothing$ and all bit strings in $\tau \sL$ are $\tau \sL$-branching.
    Define the $\sL$-\textbf{pruned*} tree $T$ of height $h$ and labels $\bm{\tau}$ to be the graph complex with edges $\|\bar{\tau} s\ket$ and vertices $|\bar{\tau} s\ket$ over branching $(\ell \le h)$ bit-strings $s$, and adjacency relation
    \begin{equation}
        \partial^{T} \|\bar{\tau} s\ket = |\bar{\tau} s\ket + |\bar{\tau} \Omega s\ket
    \end{equation}
    where \textbf{branching truncation} $\Omega s=\Omega^{\tau \sL}s $ is the longest bit string which is a truncation of $s$ but also $\tau \sL$-branching (possibly the empty bit).
\end{definition}

\begin{figure}[ht]
\centering
\subfloat[]{%
    \centering
    \includegraphics[width=0.49\columnwidth]{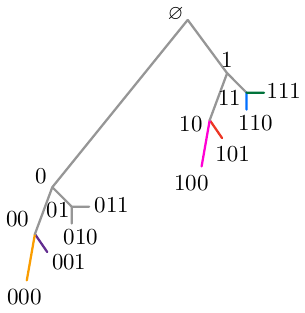}
}
\subfloat[]{%
    \centering
    \includegraphics[width=0.49\columnwidth]{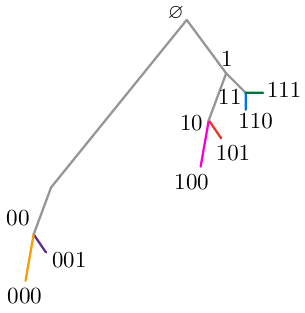}
}
\caption{Pruning*. The colored leaves form the set $\sL$. (a) denotes a binary tree with trivial labels, and the bit strings denote the vertices within the tree. (b) denotes a $\sL$-pruned* binary tree with trivial labels, and the bit strings denote the remaining vertices. In particular, the vertex $|0\ket$ in (a) is absent in (b) so that the edge $\|00\ket$ in (b) has endpoints $|00\ket$ and $|\varnothing\ket$.}
\end{figure}

\begin{remark}[Pruned* Size]
    Note that all vertices (except the root $|\varnothing\ket$) of a $\sL$-pruned* tree has degree $3$ or $1$ -- the latter corresponding to the leaves $\sL$. 
    Using the relation that for any graph $(\sV,\sE)$,
    \begin{equation}
        2|\sE| = \sum_{x\in \sV} \deg x
    \end{equation}
    and the fact that the $|\sE| = |\sV|-1$ for trees, we see that $\sL$-pruned* trees have vertices and edges of size $O(\sL)$.
\end{remark}

The crucial relation in Remark \ref{rem:crucial-relation} was the essential ingredient in defining the interpolation complex.
Since the adjacency relations for the pruned* trees replaces the truncation with $\omega\mapsto \Omega$, we must derive an analogous crucial relation, formalized in Proposition \ref{prop:crucial-relation*}.

\begin{figure}[ht]
\centering
\subfloat[]{%
    \centering
    \includegraphics[width=0.49\columnwidth]{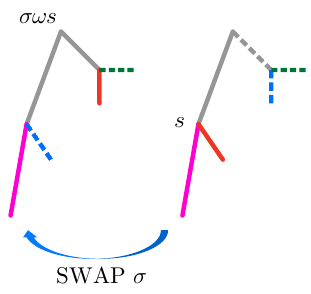}
}
\subfloat[]{%
    \centering
    \includegraphics[width=0.49\columnwidth]{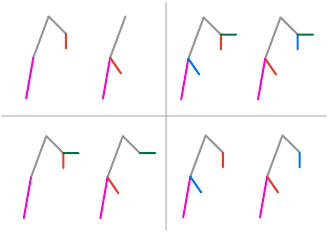}
}
\caption{SWAP and Prune*. (a) Solid edges remain after prune*, while dashed edges may or may not be pruned*. A SWAP occurs at height of string $s$. (b) shows the four concrete possibilities concisely sketched in (a).}
\label{fig:SWAP-and-prune*}
\end{figure}

\begin{lemma}[SWAP and Prune*, Fig. \ref{fig:SWAP-and-prune*}]
    \label{lem:SWAP-and-prune*}
    Let $\sL$ be leaves, $\bm{\tau}$ be a label and $\bm{\sigma}$ be a non-interacting SWAP with indices $I$. Let $s$ be a $\tau \sL$-branching $\ell$-bit string.
    \begin{itemize}
        \item If $\ell,\ell+1\notin I$, then $\sigma s$ is $\sigma \tau \sL$-branching
        \item If $\ell \in I$, then $\sigma \omega s=\omega \sigma s$ is $\sigma \tau \sL$-branching
        \item If $\ell+1\in I$, then either $\sigma s$ is $\sigma \tau \sL$-branching or $(\sigma s)b$ is $\sigma \tau \sL$-branching for some bit $b\in \{0,1\}$.
    \end{itemize}
\end{lemma}
\begin{remark}
    The second possibility ($\ell \in I$) is depicted in Fig. \ref{fig:SWAP-and-prune*}, while the third possibility ($\ell +1\in I$) is the converse of the second possibility.
\end{remark}
\begin{proof}
    Let consider the three cases separately. For notation simplicity, we omit $\tau$ and treat it as a trivial permutation. The proof is exactly the same for general $\tau$ by replacing $\sL\mapsto \tau \sL$ in the following.
    \begin{itemize}
        \item ($\ell,\ell+1\notin I$) 
        By Lemma \ref{lem:SWAP-and-prune}, we see that $\sigma(sb)=(\sigma s) b$ is a $\sigma \tau \sL$-truncation for all $b\in \{0,1\}$ and thus $\sigma s$ is $\sigma \sL$-branching.
        \item ($\ell \in I$). 
        Since $s$ is $\sL$-branching, we see that $s b$ is a $\sL$-truncation for any bit $b$. 
        Since $\ell+1\notin I$, we see that $\sigma(sb)$ is a $\sigma \sL$-truncation and thus $\omega \sigma(sb)$ is a $\sigma \sL$-truncation. However, note that $(\sigma \omega s)b=\omega \sigma (sb)$. Hence, $\sigma \omega s = \omega \sigma s$ is $\sigma \sL$-branching, where equality follows from $\ell-1\notin I$ and the crucial relation in Remark \ref{rem:crucial-relation}.
        \item ($\ell+1 \in I$).
        Since $\ell \notin I$, by Lemma \ref{lem:SWAP-and-prune}, $\sigma s$ is a $\sigma \sL$-truncation.
        In fact, there exists bit $b$, say $=0$, such that $(\sigma s)0$ is a $\sigma \sL$-truncation.
        Suppose $\sigma s$ is not $\sigma \sL$-branching. Then $(\sigma s) 1$ is not a $\sigma \sL$-truncation. 
        Since $s$ is $\sL$-branching, we see that $s0,s1$ are both $\sL$-truncations, and thus there must exist bits $b_0,b_1$ such that $s0 b_0,s1b_1$ are $\sL$-truncations. Since $\ell+2 \notin I$, by Lemma \ref{lem:SWAP-and-prune}, $\sigma(s 0b_0)=(\sigma s) b_0 0$ and $\sigma (s 1 b_1)= (\sigma s) b_1 1$ are $\sigma \sL$-truncations. 
        If $b_0=0,b_1=1$, then $(\sigma s)00$ and $(\sigma s)11$ are $\sigma \sL$-truncations so that $\sigma s$ is $\sigma\sL$-branching, and thus we reach a contradiction.
        If $b_0=1,b_1=0$, then $(\sigma s)10$ is a $\sigma \sL$-truncation which contradicts the fact that $(\sigma s)1$ is not a $\sigma \sL$-truncation.
        If $b_0=b_1=1$, then $(\sigma s)10,(\sigma s)11$ are $\sigma \sL$-truncations and thus we reach a contradiction.
        Therefore, $b_0=b_1=0$ and thus $(\sigma s)0$ is $\sigma \sL$-branching. \qedhere
    \end{itemize}
\end{proof}

\begin{figure}[ht]
\centering
\subfloat[]{%
    \centering
    \includegraphics[width=0.45\columnwidth]{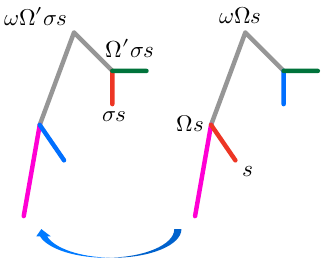}
}
\hspace{0.05\columnwidth}
\subfloat[]{%
    \centering
    \includegraphics[width=0.45\columnwidth]{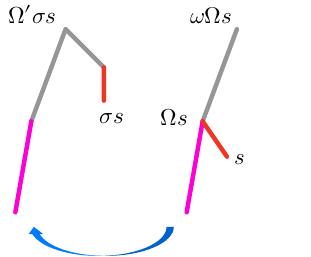}
}
\caption{Crucial Relation*. The blue arrow denotes a SWAP $\bm{\sigma}$ between the pruned* trees $T\rightarrow T'$. Bit string $s$ is a $\tau \sL$-truncation but not $\tau \sL$-branching. In particular, $s$ is not a leaf -- the colored branches should be regarded as connected to further structures. (a), (b) illustrates an example of scenario $m=m'\in I$ and $m=m'+1$ in Theorem \ref{prop:crucial-relation*}, respectively.
}
\label{fig:crucial-relation*}
\end{figure}

\begin{prop}[Crucial Relation*, Fig. \ref{fig:crucial-relation*}]
    \label{prop:crucial-relation*}
    Let $\sL$ be leaves, $\bm{\tau}$ be a label and $\bm{\sigma}$ be a SWAP with indices $I$. Let $s$ be a $\tau \sL$-truncation $\ell$-bit string.
    Let $m,m'$ be the length of $\Omega s,\sigma \Omega'\sigma s$, respectively, where $\Omega,\Omega'$ are the branching truncations with respect to $\tau \sL,\sigma \tau \sL$, respectively. Then $|m-m'|\le 1$ and
    \begin{equation}
        \sigma^{\star} \Omega s = \sigma^{\star} \sigma \Omega' \sigma s.
    \end{equation}
\end{prop}
\begin{remark}
    Similar to Remark \ref{rem:crucial-relation} in which $\sigma^{\star} \omega =\sigma^{\star} \sigma \omega \sigma$, we see that the relation still holds even for branching truncations $\Omega,\Omega'$.
\end{remark}
\begin{proof}
    For notation simplicity, we omit $\tau$ and treat it as a trivial permutation. The proof is exactly the same for general $\tau$ by replacing $\sL\mapsto \tau \sL$ in the following.
    Let $m,m'$ be the length of $\Omega s$ and $\sigma\Omega'\sigma s$, respectively.
    Note the following scenarios.
    \begin{enumerate}[label=\arabic*)]
        \item Suppose that $m,m+1\notin I$. Then $\sigma \Omega s$ must be $\sigma \sL$-branching and a truncation of $\sigma s$. Hence, $m'\ge m$, where equality implies $\Omega s=\sigma \Omega'\sigma s$
        \item Suppose that $m\in I$. 
        Then by Lemma~\ref{lem:SWAP-and-prune*} $\omega \sigma \Omega s$ must be $\sigma \sL$-branching and a truncation of $\sigma s$. Hence, $m'\ge m-1$ where equality implies $\omega \Omega s=\sigma \Omega'\sigma s$.
        \item Suppose that $m+1\in I$. Then by Lemma~\ref{lem:SWAP-and-prune*} one and only one of the following must hold.
        \begin{enumerate}[label=\alph*)]
            \item $\sigma \Omega s$ is $\sigma \sL$-branching and a truncation of $\sigma s$. Hence, $m'\ge m$, where equality implies $\Omega s=\sigma \Omega'\sigma s$
            \item $\sigma \Omega s$ is not $\sigma \sL$-branching and there exists bit $b$ such that $(\sigma \Omega s)b$ is $\sigma \sL$-branching and a truncation of $\sigma s$. Hence, $m'\ge m+1$ where equality implies $(\Omega s)b=\sigma \Omega'\sigma s$
        \end{enumerate}
    \end{enumerate}
    Group scenario 3a with 1, since they imply the same result. Note that the scenarios are disjoint and form a partition of all possibilities.
    Similar scenarios also hold for $\Omega'\sigma s$, which we denote by $1'\sqcup 3'a,2',3'b$.
    The possibilities are tabulated below, where the empty regions are not possible.
   For example, by the previous argument, if we are in scenario $1\sqcup 3a$, then $m' \ge m$. Similarly, if we are in scenario $1'\sqcup 3'a$, then $m\ge m'$. Therefore, if both scenarios occur, then $m=m'$. 
    The other possibilities are similarly derived.
    \begin{center}
    \begin{tabular}{r | c c c}
        \multicolumn{1}{c}{} & $1'\sqcup 3'a$ & $2'$ & $3'b$ \\
        \cline{2-4}
        $1\sqcup 3a$ & $m=m'$ & $m=m',m'-1$ &  \\
        $2$ & $m=m',m'+1$ & $|m-m'|\le 1$ & $m=m'+1$ \\
        $3b$ &  & $m=m'-1$ &  \\
    \end{tabular}
    \end{center}
    Based on the table, we can derive the relation between $\sigma^{\star}\Omega s$ and $\sigma^{\star}\sigma \Omega' \sigma s$ immediately for all cases so that the statement follows.
\end{proof}
\begin{definition}[Pruned* Interpolations]
    Let $\sL$ be leaves and $T',T$ be $\sL$-pruned* trees with labels $\bm{\sigma\tau}, \bm{\tau}$, respectively, where $\bm{\sigma}$ is a SWAP. 
    Define the $\sL$-\textbf{pruned*} interpolation complex $S$ between $T',T$ as the graph complex with vertices $|\bar{\tau}\sigma^{\star} s\ket$ and edges $\| \bar{\tau} \sigma^{\star} s\ket$ over $s$ such that either $s$ is branching with respect to $\tau \sL$ or $\sigma s$ is branching with respect to $\sigma \tau \sL$.
    Define
    \begin{equation}
        \partial^{S}\|\bar{\tau}\sigma^{\star} s\ket = |\bar{\tau} \sigma^{\star} s\ket + |\bar{\tau}\sigma^{\star} \tilde{\Omega} s\ket
    \end{equation}
    where \textbf{interpolating truncation} $\tilde{\Omega}$ is defined as follows: if $\Omega,\Omega'$ are the branching truncations with respect to $\tau \sL,\sigma \tau \sL$, then $\tilde{\Omega}s$ is the longer bit string among $\Omega s$ and $\sigma \Omega' \sigma s$.
\end{definition}

\begin{lemma}[Well-Definedness]
    The $\sL$-pruned* interpolation complex $S$ is well-defined.
\end{lemma}
\begin{proof}
    To make sure that $S$ is well-defined, we have to check that (1) $\tilde{\Omega}$ is well-defined and (2) the adjacency relation is well-defined.
    
    Indeed, if either $\Omega s$ or $\sigma \Omega' \sigma s$ is strictly longer bit string than the other, then $\tilde{\Omega} s$ is well-defined.
    However, in the case that the two are equally long, we must show that choosing one or the other does not make a difference, that is, $\sigma^{\star} \Omega s= \sigma^{\star} \sigma \Omega'\sigma s$.
    Indeed, this follows from Proposition \ref{prop:crucial-relation*}.
    
    Now to check that the adjacency relation is well-defined.
    Suppose edge $\|r\ket$ is given where $r$ is an $\ell$-bit string. 
    Without loss of generality, consider the case where $\ell\in I$ so that only $\omega s$ can be determined from $r=\bar{\tau} \sigma^{\star} s$.
    By definition, either $s$ is branching with respect to $\tau \sL$ or $\sigma s$ is branching with respect to $\sigma \tau \sL$.
    In the former case, by Lemma \ref{lem:SWAP-and-prune*}, we see that $\sigma \omega s = \omega \sigma s$ is branching with respect to $\sigma \tau \sL$. 
    Hence, $\sigma \Omega' \sigma s = \omega s$ must be the longer truncation so that we can choose $\tilde{\Omega} s= \omega s$ and thus well-defined from our limited information $\omega s$.
    In the latter case, by Lemma \ref{lem:SWAP-and-prune*}, we see that $\omega s$ is branching with respect to $\tau \sL$ and thus $\Omega s = \omega s$ must be the longer truncations so that we can choose $\tilde{\Omega} s =\omega s$.
    Hence, the adjacency relation is well-defined.
\end{proof}

\begin{remark}[Pruned* Interpolation Size]
    The number of vertices in $S$ is upper bounded by the sum of that of $T'$ and $T$ and thus $=O(\sL)$. This follows from the fact that the vertices in $S$ are labeled via $\bar{\tau}\sigma^{\star} s$ and thus must be $\le$ than the possibilities of $s$.
\end{remark}

\begin{figure}[ht]
\centering
\subfloat[]{%
    \centering
    \includegraphics[width=0.7\columnwidth]{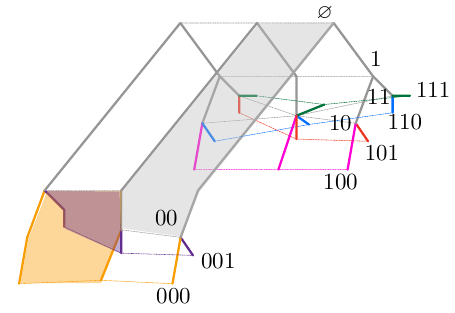}
}
\subfloat[]{%
    \centering
    \includegraphics[width=0.22\columnwidth]{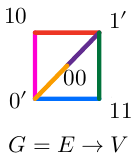}
}
\\
\subfloat[\label{fig:pruned*-cell-complex}]{%
    \centering
    \includegraphics[width=.9\columnwidth]{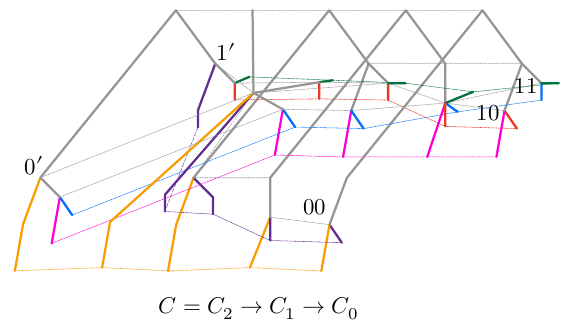}
}
\caption{Pruned* Interpolation Cone. (a) Left $T'$ and right $T$ denoted pruned* binary trees with labels $\bm{\sigma}$ and trivial labels, respectively, where $\bm{\sigma}$ is a SWAP, while the middle denotes the pruned* interpolation graph complex $S$ of $T',T$. The vertices of $T$ are labeled by $|s\ket$ with bit strings $s$.
Note that the shaded orange, purple, grey regions are all faces with 5 edges. For example, the left-most segment of the orange and purple faces denote a single edge despite drawn in a piecewise manner, while the right-most segment denotes two edges, separated at the branching point of the interpolation complex.
(b) denotes the graph $G$, which has five vertices and six edges, from which (c) is constructed as in Theorem~\ref{thm:pruned*-parsimonious-cone}. Note that (a) is the right portion of (c).
Note that the degree of any vertex in (b) is at least $2$, which guarantees that the corresponding labelled vertices in (c) are branching and thus not pruned*.
}
\label{fig:pruned*-interpolation-cone}
\end{figure}

\begin{lemma}[Pruned* Interpolations Cone, Fig. \ref{fig:pruned*-interpolation-cone}]
    \label{lem:pruned*-interpolation-cone}
    Let $\sL$ be leaves at height $h$.
    Let $T',T$ be $\sL$-pruned* height $h$ trees with labels $\bm{\sigma\tau}, \bm{\tau}$ and truncations $\Omega',\Omega$, respectively, where $\bm{\sigma}$ is a SWAP.
    Let $S$ denote the $\sL$-pruned* interpolation with truncation $\tilde{\Omega}$.
    Let $T'S,ST$ be $\sL$-pruned* height $h$ binary trees with labels $\bm{\sigma\tau},\bm{\tau}$, respectively.
    Define $g_i:(T'S \oplus ST)_i \to (T'\oplus S\oplus T)_i$ for $i=1,0$ via
    \begin{align}
        g_1\| \bar{\tau} \sigma s, T'S\ket &= \| \bar{\tau} \sigma s, T'\ket + \| \bar{\tau} \sigma^{\star} \sigma s, S\ket \nonumber\\
        &\; + \|\bar{\tau} \sigma^{\star}  \Omega \sigma s,S\ket \mathbbm{1}\{\|\Omega \sigma s\| > \|\Omega' s\|\}\\
        g_1\| \bar{\tau} s, ST\ket &= \| \bar{\tau} \sigma^{\star} s, S\ket + \| \bar{\tau} s, T\ket \\
        &\; + \|\bar{\tau}\sigma^{\star} \Omega' s, S\ket \mathbbm{1}\{\|\Omega's\| > \|\Omega \sigma s\|\}
    \end{align}
    where $\|\cdot\|$ denotes the length of the bit string\footnote{Note that this is different from Hamming weight.} and $\mathbbm{1}\{\cdot\}$ denotes an indicator function.
    Similarly, define $g_0$ with $\|\cdots \ket$ replaced by $|\cdots\ket$ and without the third indicator term.
    Then $g$ is a chain map and the cone complex $C=\cone (g)$ -- see Fig. \ref{fig:pruned*-interpolation-cone} -- satisfies
    \begin{align}
        \label{eq:pruned*-cone-prop1}
        H_2(C)=H_1(C)=0, &\quad H_0(C) \cong \dF_2\\
        \label{eq:pruned*-cone-prop2}
        C_2 \xrightleftharpoons[4]{5} C_1 \xrightleftharpoons[9]{2} C_0.
    \end{align}
\end{lemma}

\begin{proof}
    Since the adjacency relations of the $\sL$-pruned* trees and interpolations are slightly different from Eq.~\eqref{eq:tree-adjacency} and Eq.~\eqref{eq:interpolation-adjacency}, we will need to slightly modify the proof of Lemma \ref{lem:pruned-interpolation-cone}.
    For the sake of notation simplicity, assume that $\tau$ is the trivial permutation, since the proof is exactly the same otherwise by replacing $\sL\mapsto \tau \sL$.
    Indeed, note that $g_i$ are well-defined maps and that
    \begin{align}
        \partial^{T'\oplus S} g_1 \| \sigma s,T'S\ket &= | \sigma s,T'\ket +| \sigma \Omega' s,T'\ket \nonumber\\
        &+ | \sigma^{\star} \sigma s,S\ket +| \sigma^{\star} \tilde{\Omega} \sigma s,S\ket \nonumber\\
        &+ \left( |\sigma^{\star} \Omega \sigma s,S\ket +|\sigma^{\star} \tilde{\Omega} \Omega \sigma s,S\ket \right) \nonumber\\
        &\quad\quad \times1\{\|\Omega \sigma s\| > \|\Omega' s\|\}.
    \end{align}
    Moreover,
    \begin{align}
        g_0 \partial^{T'\oplus S} \| \sigma s, T'S\ket &= |\sigma s,T'\ket +|\sigma \Omega' s, T'\ket \nonumber\\
        &+|\sigma^{\star} \sigma s,S\ket + |\sigma^{\star} \sigma \Omega' s,S\ket.
    \end{align}
    It's clear that if $\|\Omega \sigma s\| \le \|\Omega' s\|$, then we can choose $\tilde{\Omega} \sigma s = \sigma \Omega' s$ and thus the two equations are equal.
    Hence, to show that $g$ is a chain map, we shall focus on the case where $\|\Omega \sigma s\| > \|\Omega' s\|$.
    
    We can choose $\tilde{\Omega}\sigma s=\Omega \sigma s$ and thus it's sufficient to prove that $\sigma^{\star} \tilde{\Omega} \Omega \sigma s = \sigma^{\star} \sigma \Omega' s$.
    Indeed, by Proposition \ref{prop:crucial-relation*}, we see that $\Omega' s = \omega \sigma \Omega \sigma s$ and that $\|\Omega \sigma s\|\in I$.
    Hence, $\sigma^{\star} \sigma \Omega' s =\omega \Omega \sigma s$. 
    Conversely, to determine $\tilde{\Omega} \Omega \sigma s$, consider $\Omega \Omega \sigma s$ and $\sigma \Omega'\sigma \Omega \sigma s$.
    Since $\|\Omega \sigma s\| \in I$, by Lemma \ref{lem:SWAP-and-prune*}, we see that $\omega \sigma \Omega \sigma s$ is $\sigma \sL$-branching and thus $\Omega' \sigma \Omega \sigma s = \omega \sigma \Omega \sigma s$ must be the (not necessarily strictly) longer of the two ($\Omega \Omega \sigma s$ vs. $\sigma \Omega' \sigma \Omega \sigma s$) so that $\tilde{\Omega} \Omega \sigma s=\sigma \omega \sigma \Omega \sigma s = \omega \Omega \sigma s$.
    Hence, $\sigma^{\star} \tilde{\Omega} \Omega \sigma s = \sigma^{\star} \sigma \Omega' s$ and that $g_i$ is a chain map.

    The rest of the argument follow similarly as in Lemma \ref{lem:pruned-interpolation-cone}, though the weight in Eq.~\eqref{eq:pruned-cone-prop2} is slightly different due to the extra term in the chain map $g$.
\end{proof}

\begin{theorem}[Pruned* Parsimonious Cone]
    \label{thm:pruned*-parsimonious-cone}
    Let $G=E\to V$ denote the complex of a bipartite graph with vertex partitions $\sV=\sV_0\sqcup \sV_1$ where $|\sV_i| \le 2^{h_i}$, and each vertex of $G$ has degree $\ge 2$ with max degree $\Delta(G)$.
    Then there exists cell complex $C$ such that a $3h$-deformation of $G$ with congestion $\Delta(G)$ is a subgraph of $C$ with 
    \begin{equation}
        \dim C_0 = O(|\sE| \log(|\sV||\sE|))
    \end{equation}
    and 
    \begin{align}
        H_2(C)=H_1(C)=0, &\quad H_0(C) \cong \dF_2\\
        C_2 \xrightleftharpoons[4]{5} C_1 \xrightleftharpoons[9]{2} C_0.
    \end{align}
    In particular, if $\Delta(G)=O(1)$, then 
    \begin{equation}
        \dim C_0 = O(|\sV| \log|\sV|).
    \end{equation}
\end{theorem}
\begin{proof}
    The proof is similar to Theorem \ref{thm:pruned-parsimonious-cone}, except that the overhead is smaller and the weights in Eq.~\eqref{eq:pruned*-cone-prop2} are slightly different due to the extra term in the chain map $g$.
    Also note the additional requirement that each vertex has degree $\ge 2$. This is so that the vertices correspond to branching points and thus won't be pruned*.
\end{proof}

\begin{corollary}
    \label{cor:pruned*-parsimonious-cone}
    Let $G$ denote a graph complex with max degree $\Delta(G)=O(1)$.
    Then there exists cell complex $C$ such that a $6h$-deformation of $G$ with congestion $\Delta(G)$ is a subgraph of $C$ with 
    \begin{equation}
        \dim C_0 = O(|\sV| \log|\sV|)
    \end{equation}
    And 
    \begin{align}
        H_2(C)=H_1(C)=0, &\quad H_0(C) \cong \dF_2\\
        C_2 \xrightleftharpoons[4]{5} C_1 \xrightleftharpoons[9]{2} C_0
    \end{align}
\end{corollary}
\begin{proof}
    Similar to Corollary \ref{cor:pruned-parsimonious-cone}, add edges while keeping the degree bounded so that the condition in Theorem \ref{thm:pruned*-parsimonious-cone} is satisfied. The statement then follows.    
\end{proof}

\subsection{Direct Embedding}
\label{sec:direct-embedding-instead-of-subdivision}
\begin{figure}[ht]
\centering
\subfloat[\label{fig:cellulation-intuition}]{%
    \centering
    \includegraphics[width=0.4\columnwidth]{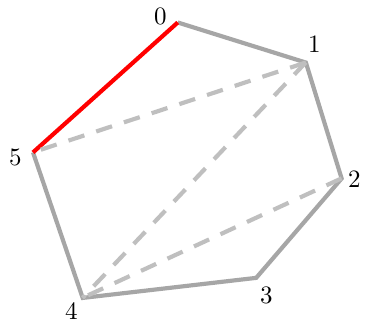}
}
\subfloat[\label{fig:cellulation-adjacency}]{%
    \centering
    \includegraphics[width=0.4\columnwidth]{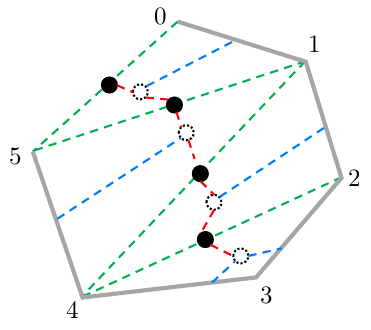}
}
\caption{Cellulation. (a) The grey lines indicate the path $\gamma e$, while the red line denotes $e$. The dashed lines indicate the added edges, which cellulate the simple cycle. (b) denote the adjacency relations of the dangling repetition code $R^{\multimap}$ and the original subdivision $\gamma e$ as depicted by the diagram in Eq.~\eqref{eq:cellulation-diagram}.
}
\label{fig:cellulation}
\end{figure}

As stated in Theorems~\ref{thm:parsimonious-cone}, \ref{thm:pruned-parsimonious-cone}, \ref{thm:pruned*-parsimonious-cone} or their Corollaries, given a graph complex $G$, there exists a parsimonious cone which contains a deformation of $G$.
In this section, we modify $C\mapsto C'$ so that $G$ is a subgraph of $C'$ instead.
To do this, let's start simple and consider the deformation $\gamma (e)$ of a single edge $e$.
As depicted in Fig. \ref{fig:cellulation-intuition}, insert the original edge $e$ to the endpoint of $\gamma e$ so that $\gamma e +e$ forms a simple cycle, e.g., $\gamma e$ consists of vertices and edges involved in the sequence $(0,1,...,n)$ while $e$ corresponds to $(n,0)$.
Cellulate the cycle via additional edges and faces.
The general process is formalized by the following lemma


\begin{lemma}[Cellulation]
    Let $G:E\to V$ denote a graph complex of max degree $\Delta(G)$ and $C$ denote a cell complex with max degree $\Delta(C)$ such that a deformation $\gamma G$ with congestion $\Delta(G)$ is a subgraph of $C$. 
    Then there exists cell complex $C'$ such that $G$ is a subgraph of $C'$ and that
    \begin{equation}
        H_i(C') \cong H_i(C),\forall i\quad \dim C_0 =\dim C_0'.
    \end{equation}
    Moreover, if $C$ has max weights denoted by
    \begin{equation}
        C_2 \xrightleftharpoons[\qubit]{\weight} C_1 \xrightleftharpoons[\Delta(C)]{2} C_0,
    \end{equation}
    then $C'$ has max weight denoted by
    \begin{equation}
        C_2' \xrightleftharpoons[\qubit+\Delta(G)]{\max(\weight,3)} C_1' \xrightleftharpoons[\Delta(C)+\Delta(G)]{2} C_0'.
    \end{equation}
\end{lemma}
\begin{proof}
    Fix an edge $e\in \sE$ and let $\|\gamma e\|$ be length of the deformed edge.
    Let $R^{\multimap}=R^{\multimap}(e)$ denote the dangling repetition code on $\|\gamma e\|-1$ bits, i.e., $R^{\multimap}_1 \to R_{0}^{\multimap}$ has 1-cells $\|i\ket=\|i;e\ket$ and 0-cells $|i\ket=|i;e\ket$ for $i=1,...,\|\gamma e\|-1$ , with adjacency relation
    \begin{equation}
        \partial^{\multimap} \|i\ket = |i\ket+|i+1\ket
    \end{equation}
    where $|i\ket$ is implicitly $=0$ if $|i\ket$ is not well-defined.
    Note that $R^{\multimap}$ is not a graph complex and that $H_1(R^{\multimap})=H_0(R^{\multimap}) =0$.

    Without loss of generality, label the vertices of path $\gamma e$ via $0,...,\|\gamma e\|$ and consider the linear map $g_i =g_i(e):R_i^{\multimap} \to C_i$ shown in the following diagram and, correspondingly, Fig. \ref{fig:cellulation-adjacency}.
    \begin{equation}
    \label{eq:cellulation-diagram}
    \begin{tikzpicture}[baseline]
    \matrix(a)[matrix of math nodes, nodes in empty cells, nodes={minimum size=25pt},
    row sep=1.5em, column sep=2em,
    text height=1.25ex, text depth=0.25ex]
    {& R_1^{\multimap}  & R_0^{\multimap} \\
     C_2  & C_1 & C_0 \\};
    \path[->,font=\scriptsize]
    (a-1-2) edge node[above]{$\partial^{\multimap}$} (a-1-3)
    (a-2-1) edge node[above]{$\partial^{C}$} (a-2-2)
    (a-2-2) edge node[above]{$\partial^{C}$} (a-2-3);
    \path[->,font=\scriptsize]
    (a-1-2) edge node[right]{$g_1$} (a-2-2)
    (a-1-3) edge node[right]{$g_0$} (a-2-3);
    \end{tikzpicture}
    \end{equation}
    Specifically, $g_0$ is defined to map $|i\ket$ to vertices in $\gamma e$ labeled by $\lfloor i/2\rfloor$ and $n-\lceil i/2\rceil +1$
    and $g_1$ maps $\|i\ket$ to edge(s) in $\gamma e$ with endpoints
    \begin{itemize}
        \item $\ell,\ell+1$ if $i=2\ell+1\ne \|\gamma e\| -1$
        \item $n-\ell+1,n-\ell$ if $i=2\ell \ne \|\gamma e\| -1$ 
        \item $\lceil n/2\rceil-1, \lceil n/2\rceil$ and $\lceil n/2\rceil, \lceil n/2\rceil+1$ for $i=\|\gamma e\| -1$
    \end{itemize}
    It's straightforward to check that $g$ is a chain map. 
    Perform the described operation for all edges $e\in \sE$. Specifically, let 
    \begin{equation}
        R^{\multimap} \mapsto \bigoplus_{e} R^{\multimap}(e), \quad g_i \mapsto \sum_{e} g_i(e)
    \end{equation}
    Then it's clear that $g_i$ is a chain map since it's a sum of chain maps. Let $C'=\cone(g)$ and thus by the Snake Lemma, $H_i(C') \cong H_i(C)$ for all $i$.
    Identify $e\in \sE$ with $|0;e\ket\in R^{\multimap}$ and thus it's clear that $G$ is a subgraph of $C'$.
    
    Finally, note that the weight of the cone complex $C'$ is obtained via the following diagram, where the number denotes the maximum column weight of the corresponding map
    \begin{equation}
    \label{eq:weight-diagram-cellulation}
    \begin{tikzpicture}[baseline]
    \matrix(a)[matrix of math nodes, nodes in empty cells, nodes={minimum size=25pt},
    row sep=2em, column sep=2em,
    text height=1.25ex, text depth=0.25ex]
    {& R_1^{\multimap}  & R_0^{\multimap} \\
     C_2  & C_1 & C_0 \\};
    \path[-left to,font=\scriptsize,transform canvas={yshift=0.2ex}]
    (a-1-2) edge node[above]{$2$}  (a-1-3)
    (a-2-1) edge node[above]{$\weight$}  (a-2-2)
    (a-2-2) edge node[above]{$2$}  (a-2-3);
    \path[left to-,font=\scriptsize,transform canvas={yshift=-0.2ex}]
    (a-1-2) edge node[below]{$2$}  (a-1-3)
    (a-2-1) edge node[below]{$\qubit$}  (a-2-2)
    (a-2-2) edge node[below]{$\Delta(C)$}  (a-2-3);
    \path[-left to,font=\scriptsize,transform canvas={xshift=0.2ex}]
    (a-1-2) edge node[right]{$2$}  (a-2-2)
    (a-1-3) edge node[right]{$2$}  (a-2-3);
    \path[left to-,font=\scriptsize,transform canvas={xshift=-0.2ex}]
    (a-1-2) edge node[left]{$\Delta(G)$}  (a-2-2)
    (a-1-3) edge node[left]{$\Delta(G)$}  (a-2-3);
    \end{tikzpicture}
    \end{equation}
    In particular, note that the map $C_1 \mapsto R_1^{\multimap}$ has max weight $\Delta(G)$ since $\gamma G$ has congestion $\Delta(G)$.
\end{proof}

Combine the cellulation lemma with Theorem \ref{thm:pruned*-parsimonious-cone} or Corollary \ref{cor:pruned*-parsimonious-cone} to obtain the following.
\begin{theorem}[Cellulated Cone]
    \label{thm:cellulated-cone}
    Let $G=E\to V$ denote a graph complex with bounded degree $\Delta = \Delta(G)$.
    Then there exists cell complex $C$ such that $G$ is a subgraph of $C$ with 
    \begin{equation}
        \dim C_0 = O(|\sV| \log|\sV|),
    \end{equation}
    and 
    \begin{align}
        H_2(C)=H_1(C)=0, &\quad H_0(C) \cong \dF_2\\
        \label{eq:reduced-weights-ancilla}
        C_2 \xrightleftharpoons[4+\Delta]{5} C_1 \xrightleftharpoons[9+\Delta]{2} C_0.
    \end{align}
\end{theorem}

\begin{remark}[Minor Difference]
    \label{rem:minor-difference}
    If we compare Eq.~\eqref{eq:reduced-weights-ancilla} with the weights of the parsimonious cone claimed in the proof of Lemma 3.4 of Ref.~\cite{hsieh2025simplified} (end of page 9), we see that the weights are exactly the same except for an additional term accounting for the max degree $\Delta$ of the original graph.
    This is due to the possibility that neither the original graph nor its edge-subdivision is necessarily embedded in the parsimonious cone; instead, only a deformation of the original graph is embedded.
    For example, as depicted in Fig. \ref{fig:pruned*-cell-complex}, the pink and blue branches in the left-most binary tree overlap when traversing the path to vertex $0'$, and thus when cellulating the cone, an edge can participate in the cellulation of $\Delta$ many cycles.
    This contribution was overlooked in Ref.~\cite{hsieh2025simplified} and thus may result in slightly larger reduced weights than originally claimed.
\end{remark}

\section{Attaching Ancilla Cone}

\label{sec:attaching-ancilla-cone}

In this section we describe how to construct and attach an ancilla system based on a parsimonious cone to perform a fault-tolerant logical measurement. 

Let $D=D_2 \to D_1 \to D_0$ denote the data LDPC CSS code\footnote{We note that the measurement can be performed on non-CSS stabilizer codes as well, see Appendix~\ref{app:other_schemes}.} with (arbitrary) convention that $D_2,D_1,D_0$ denote the $Z$-type checks, qubits and $X$-type checks, respectively, and has weights denoted by
\begin{equation}
    D_2 \xrightleftharpoons[\qubit_Z]{\weight_Z} D_1 \xrightleftharpoons[\weight_X]{\qubit_X} D_0.
\end{equation}
Without loss of generality (and to keep notation consistent), consider measuring $X$-type logicals, which implies considering the cocomplex $D^\top =D_0 \to D_1 \to D_2$ with codifferential $\delta^D$.
Specifically, let $\ell\in D_1$ denote an ($X$-type) logical operator and define the measurement graph complex $G$ of $\ell$ as follows.
Let $\sV$ denote vertices which are in 1:1 correspondence with qubits (basis elements in $D_1$) which $\ell$ acts on.
Note that every $Z$-type check (basis element $z\in D_0$) which overlaps with $\ell$ must have an even overlap, i.e., $|\delta^D z\cap \ell|$ must be even.
Therefore, the overlap can be paired up arbitrarily to form edges of the graph complex.
This leads to a degree bounded in $\Delta(G)=\weight_Z \qubit_Z$.
Further add arbitrary edge while keeping the degree bounded by $\weight_Z\qubit_Z$ until the induced graph is expanding\footnote{The inclusion of an expander graph is sufficient to guarantee that the code distance of the deformed code is preserved. However, in practice, it is commonly unnecessary, e.g., \cite{cowtan2024ssip,cross2024improved,williamson2024low,swaroop2026universal,yoder2025tour}.}, i.e., the Cheeger constant $h(G) =\Omega(1)$ is bounded below.\footnote{This can be achieved by adding random edges, or using more sophisticated methods~\cite{lubotzky1988ramanujan,hoory2006expander,alon2008elementary}.}
In particular, the construction induces a chain map $f^\top$ from the graph cocomplex $G^\top$ to the data cocomplex $D^\top$, where $f_0^\top$ maps each vertex to the corresponding qubit in $\ell$, and $f_1^\top$ maps the edge $e\in \sE$ which correspond to a $Z$-check $z$ overlapping with $\ell$ to $z$; conversely, arbitrarily added edges are mapped to zero.
\begin{equation}
    \begin{tikzpicture}[baseline]
    \matrix(a)[matrix of math nodes, nodes in empty cells, nodes={minimum size=25pt},
    row sep=1.5em, column sep=2em,
    text height=1.25ex, text depth=0.25ex]
    {& V  & E \\
     D_0  & D_1 & D_2 \\};
    \path[->,font=\scriptsize]
    (a-1-2) edge node[above]{$\delta^{G}$} (a-1-3)
    (a-2-1) edge node[above]{$\delta^{D}$} (a-2-2)
    (a-2-2) edge node[above]{$\delta^{D}$} (a-2-3);
    \path[->,font=\scriptsize]
    (a-1-2) edge node[right]{$f_0^\top$} (a-2-2)
    (a-1-3) edge node[right]{$f_1^\top$} (a-2-3);
    \end{tikzpicture}
\end{equation}

By Theorem \ref{thm:cellulated-cone}, we can replace the graph (co)complex $G^\top$ with the cellulated cone (co)complex $C^\top$. 
Because $G$ is a subgraph of $C$, the inclusion map $\iota: G \hookrightarrow C$ induces an injective chain map, and thus a surjective map of cocomplexes $\pi: C^\top \rightarrow G^\top$ such that $\pi = \iota^\top$. As a commutative diagram,

\begin{equation}
    \begin{tikzpicture}[baseline]
    \matrix(a)[matrix of math nodes, nodes in empty cells, nodes={minimum size=25pt},
    row sep=1.5em, column sep=2em,
    text height=1.25ex, text depth=0.25ex]
    {C_0  & C_1 & C_2 \\
     V  & E & \\};
    \path[->,font=\scriptsize]
    (a-1-1) edge node[above]{$\delta^{C}$} (a-1-2)
    (a-1-2) edge node[above]{$\delta^{C}$} (a-1-3)
    (a-2-1) edge node[above]{$\delta^{G}$} (a-2-2);
    \path[->,font=\scriptsize]
    (a-1-1) edge node[right]{$\pi_0$} (a-2-1)
    (a-1-2) edge node[right]{$\pi_1$} (a-2-2);
    \end{tikzpicture}
    \label{eq:graph_projection}
\end{equation}

and so there is a map between cochain complexes $g^\top: C^\top \rightarrow D^\top$, where $g_i^\top = f_i^\top\pi_i$.
\begin{equation}
    \label{eq:logical-measurement-diagram}
    \begin{tikzpicture}[baseline]
    \matrix(a)[matrix of math nodes, nodes in empty cells, nodes={minimum size=25pt},
    row sep=1.5em, column sep=2em,
    text height=1.25ex, text depth=0.25ex]
    {& C_0  & C_1 & C_2 \\
     D_0  & D_1 & D_2 &\\};
    \path[->,font=\scriptsize]
    (a-1-2) edge node[above]{$\delta^{C}$} (a-1-3)
    (a-1-3) edge node[above]{$\delta^{C}$} (a-1-4)
    (a-2-1) edge node[above]{$\delta^{D}$} (a-2-2)
    (a-2-2) edge node[above]{$\delta^{D}$} (a-2-3);
    \path[->,font=\scriptsize]
    (a-1-2) edge node[right]{$g_0^\top$} (a-2-2)
    (a-1-3) edge node[right]{$g_1^\top$} (a-2-3);
    \end{tikzpicture}
\end{equation}

\begin{theorem}[Deformed Code]
    \label{thm:deformed-code}
    Let $\cone(g^\top) =(\cone g)^\top$ be that obtained in Eq.~\eqref{eq:logical-measurement-diagram}.
    Then $M=\cone(g)$ is LDPC with $H_1(M^\top) \cong H_1(D^\top)/[\ell]$ where $[\ell]\in H_1(D^\top)$ is the equivalence class induced by the logical $\ell$. 
    Moreover, if $h(G)$ is the Cheeger constant of $G$, then the ($X$-type) code distance $d(M^\top)$ of $M^\top$ is
    \begin{equation}
        d(M^\top) \ge \min(h(G),1) d(D^\top)
    \end{equation}
    While the ($Z$-type) code distance $d(M)$ of $M$ is
    \begin{equation}
        d(M) \ge d(D)
    \end{equation}
\end{theorem}

\begin{proof}
    The isomorphism of logicals follows from the Snake Lemma, or equivalently, the framework in Ref. \cite{yuan2025unified} via padding the diagram with zeros as follows and viewing the column code $H_0(C^\top) \to H_1(D^\top) \to 0$.
    \begin{equation}
    \begin{tikzpicture}[baseline]
    \matrix(a)[matrix of math nodes, nodes in empty cells, nodes={minimum size=25pt},
    row sep=1.25em, column sep=1.5em,
    text height=1.25ex, text depth=0.25ex]
    {&& C_0  & C_1 & C_2 \\
     & D_0  & D_1 & D_2 & \\
     0 & 0 & 0 &&\\};
    \path[->,font=\scriptsize]
    (a-1-3) edge (a-1-4)
    (a-1-4) edge (a-1-5)
    (a-2-2) edge (a-2-3)
    (a-2-3) edge (a-2-4)
    (a-3-1) edge (a-3-2)
    (a-3-2) edge (a-3-3);
    \path[->,font=\scriptsize]
    (a-1-3) edge node[right]{$g_0$} (a-2-3)
    (a-1-4) edge node[right]{$g_1$} (a-2-4)
    (a-2-2) edge (a-3-2)
    (a-2-3) edge (a-3-3);
    \end{tikzpicture}
    \end{equation}
    To obtain the weights, consider the following diagram, where the number denotes the maximum column weight of the corresponding map
    \begin{equation}
    \label{eq:weight-diagram-modifed-code}
    \begin{tikzpicture}[baseline]
    \matrix(a)[matrix of math nodes, nodes in empty cells, nodes={minimum size=25pt},
    row sep=2em, column sep=3em,
    text height=1.25ex, text depth=0.25ex]
    {& C_0  & C_1 & C_2 \\
     D_0  & D_1 & D_2 &\\};
    \path[-left to,font=\scriptsize,transform canvas={yshift=0.2ex}]
    (a-1-2) edge node[above]{$9+\weight_Z \qubit_Z$}  (a-1-3)
    (a-1-3) edge node[above]{$4+\weight_Z \qubit_Z$}  (a-1-4)
    (a-2-1) edge node[above]{$\weight_X$}  (a-2-2)
    (a-2-2) edge node[above]{$\qubit_Z$}  (a-2-3);
    \path[left to-,font=\scriptsize,transform canvas={yshift=-0.2ex}]
    (a-1-2) edge node[below]{$2$}  (a-1-3)
    (a-1-3) edge node[below]{$5$}  (a-1-4)
    (a-2-1) edge node[below]{$\qubit_X$}  (a-2-2)
    (a-2-2) edge node[below]{$\weight_Z$}  (a-2-3);
    \path[-left to,font=\scriptsize,transform canvas={xshift=0.2ex}]
    (a-1-2) edge node[right]{$1$}  (a-2-2)
    (a-1-3) edge node[right]{$1$}  (a-2-3);
    \path[left to-,font=\scriptsize,transform canvas={xshift=-0.2ex}]
    (a-1-2) edge node[left]{$1$}  (a-2-2)
    (a-1-3) edge node[left]{$1$}  (a-2-3);
    \end{tikzpicture}
    \end{equation}
    In particular, this implies that
    \begin{equation}
        M_2 \xrightleftharpoons[5+\weight_Z \qubit_Z]{\max(\weight_Z+1,5)} M_1 \xrightleftharpoons[\max(\weight_X,10+\weight_Z\qubit_Z)]{\qubit_X+1} M_0
    \end{equation}
    Hence, if $D$ is LDPC, then so is $M$.
    
    To show that the code distance $d(M^\top)$ is lower-bounded, by the Cleaning Lemma of Ref. \cite{yuan2025unified}, it's sufficient to show that for any $s\in C_0$, there exists $\sigma \in C_0$ such that $\delta^C s=\delta^C \sigma$ with
    \begin{equation}
        |\delta^C \sigma| \ge \min (h(G),1) |g \sigma|.
    \end{equation}
    Indeed, let $\pi$ denote the projection chain map (transpose of embedding) from $C \to G$, as in Eq.~\eqref{eq:graph_projection}. Then
    \begin{align}
        |\delta^C s| &\ge |\pi \delta^C s| \\
        &= |\delta^G \pi s|.
        \label{eq:projection_inequality}
    \end{align}
    By definition of the Cheeger constant, if we overload notation and use $\sV$ to denote the element of $V$ which is equal to the sum of all basis elements (similarly overload the notation $\sC_0$), then 
    \begin{equation}
        |\delta^G \pi s| \ge h(G) \min(|\pi s|,|\pi s+\sV|).
    \end{equation}
    Note that $\sV=\pi \sC_0$. 
    Since $\sC_0\in \ker \delta^C$, we see that we can set $\sigma=s$ or $s+\sC_0$ depending on whether $|\pi s|$ or $|\pi s+\sV|$ obtains the minimum, and thus
    \begin{equation}
        |\delta^C s|\ge h(G) |\pi \sigma| = h(G) |g \sigma|.
    \end{equation}
    Similarly, to lower bound the code distance $d(M)$, consider the following diagram by reversing arrows. It's then clear that the sufficient condition for the Cleaning Lemma is trivially satisfied with coefficient $=1$ and thus $d(M)\ge d(D)$.
    \begin{equation}
    \begin{tikzpicture}[baseline]
    \matrix(a)[matrix of math nodes, nodes in empty cells, nodes={minimum size=25pt},
    row sep=1.25em, column sep=1.5em,
    text height=1.25ex, text depth=0.25ex]
    {&& 0  & 0 & 0 \\
     & D_2  & D_1 & D_0 & \\
     C_2 & C_1 & C_0 &&\\};
    \path[->,font=\scriptsize]
    (a-1-3) edge (a-1-4)
    (a-1-4) edge (a-1-5)
    (a-2-2) edge (a-2-3)
    (a-2-3) edge (a-2-4)
    (a-3-1) edge (a-3-2)
    (a-3-2) edge (a-3-3);
    \path[->,font=\scriptsize]
    (a-2-2) edge node[right]{$g_1$} (a-3-2)
    (a-2-3) edge node[right]{$g_0$} (a-3-3)
    (a-1-3) edge (a-2-3)
    (a-1-4) edge (a-2-4);
    \end{tikzpicture}
    \end{equation}
    We conclude that the deformed code has all the desired properties.
\end{proof}

\begin{remark}
    In Eq.~\eqref{eq:projection_inequality} note that $|\delta^C s| \geq |\delta^G \pi s|$, but it is not generally the case that $|\delta^C s| = |\delta^G \pi s|$. Hence it is possible for the ancilla cone $C$ to have increased relative expansion~\cite{swaroop2026universal} when compared to the original graph $G$. As a consequence, in practical scenarios $G$ could be constructed with a lower Cheeger constant than 1, which is still sufficient to preserve the distance $d(M^\top)$ regardless of the logical structure of the original code $D$ so long as the relative expansion of $C$ is at least 1.
\end{remark}

\section{Discussion}

\label{sec:discuss}


In this work, we described the construction of an ancilla system for quantum code surgery that has asymptotic qubit overhead $O(\Weight \log \Weight)$ for an operator of weight $\Weight$. 
This overhead was achieved by replacing the decongestion and thickening procedure, introduced by Freedman and Hastings~\cite{freedman2021building}, with the parsimonious coning procedure, introduced in Ref.~\cite{berdnikov2022parsimonious} and used for quantum weight reduction in Ref.~\cite{hsieh2025simplified}. 
This ancilla system can be incorporated to improve the asymptotic overhead of a range of quantum code surgery procedures and architectures, as summarized in Table~\ref{tab:tabulated-improvement}.

Our results push the asymptotic performance of quantum code surgery procedures beyond the previous state-of-the-art. 
However, a number of open questions remain. 
Is it possible to prove that the asymptotic overhead of our quantum code surgery procedure is the lowest achievable for any general purpose quantum code surgery procedure based on an auxiliary graph? 
Is it possible to improve the overhead of quantum code surgery to $O(\Weight)$ by incorporating more of the structure of the underlying code? 
Is it possible to improve the overhead of parallel, or fast (in the sense of single-shot or constant-time), quantum code surgery by applying the quantum weight reduction techniques to many logical operators simultaneously? 
Alternatively, what are the fundamental lower bounds on the spacetime volume required to fault-tolerantly measure a set of commuting logical operators in a qLDPC code? 
These are well-motivated questions which we hope future works could explore. 

\section*{Acknowledgements}
AC, ZH, and DJW thank Ted Yoder for insightful discussions. 
DJW thanks Harry Putterman for useful discussions. 
ACY is employed by Iceberg Quantum, and was also supported by the Laboratory for Physical Sciences at CMTC in University of Maryland, College Park.
AC is employed by Xanadu, and would like to thank the team there for helpful discussions.
ZH acknowledges support from the MIT Department of Mathematics, the MIT-IBM Watson AI Lab, and the NSF Graduate Research Fellowship Program under Grant No. 2141064.
DJW is supported by the Australian Research Council Discovery Early Career Research Award (DE220100625). 
This work was initiated while ZH, TCL, and DJW were visiting the Simons Institute for the Theory of Computing. 

\bibliography{main.bbl}

\appendix
\section{Overhead Reduction for Various Surgery Schemes}\label{app:other_schemes}

In the main text we focused on measuring a single $\bar{Z}$ or $\bar{X}$ logical operator with weight $\Weight$ in a CSS code, and reduced the qubit overhead   from $\mathcal{O}(\Weight \log^3 \Weight)$ to $\mathcal{O}(\Weight \log \Weight)$.

Here, we expand on the claim that our method has a knock-on effect which reduces the asymptotic overheads of other logical measurement schemes: universal adapters~\cite{swaroop2026universal}, parallel logical measurement~\cite{cowtan2025parallel}, and fixed-connectivity architectures~\cite{he2025extractors}.

First we describe a mild generalization of the scheme: from measuring CSS-type operators in any CSS code to measuring arbitrary Pauli operators in any stabilizer code. This generalization follows along very similar lines as in Ref.~\cite{williamson2024low}.

\subsection{Measurement of Arbitrary Pauli Operators in Stabilizer Codes}

In the scheme in the main body, cf.~\ref{sec:measurement-graph}, the primary reason to introduce new data qubits to the ancilla system is to ensure that the new $X$ checks do not anticommute with the $Z$ checks in the pre-existing code. We must ensure that similar steps are taken when moving from the CSS case to the more general case.

\subsubsection{$X$-type operator in a stabilizer code}\label{sec:X_type_in_stab}

Initially, consider a stabilizer code $D$ that is non-CSS, but where the logical to be measured is still of the form $X(\ell)$ for some set of qubits $\ell \in \dF_2^n$. A single-qubit Pauli operator $X(i)$ anticommutes with checks that have support in the $Y$ or $Z$ basis on $i$. Given a check $c$, let $\mathcal{P}(c)$ denote its Pauli action on qubits, and let $\mathcal{P}(c) \restriction_i$ be the Pauli action of $c$ on qubit $i$.

We must consider each of the stabilizer generators $c$ such that $\mathcal{P}(c) \restriction_i \in \{Y, Z\}$ for any $i \in {\rm supp}(\ell)$, as each new $X$ check would, naively, anticommute with any such check $c$. Let $c^{Y, Z}$ be the support of $c$ on data qubits where it acts only with a $Y$ or $Z$. Note that in order for $X(\ell)$ to be a logical operator, we have
\begin{equation}
    |c^{Y, Z} \cap \ell| = 0 \mod 2, \quad \forall c.
    \label{eq:check_commute}
\end{equation} 

The approach is as follows. Introduce the measurement graph $\mathcal{G}$ with a set of vertices $V$ in bijection with elements in ${\rm supp}(\ell)$. The vertices $V$ will be $X$-checks in the ancilla system, connected 1-to-1 to the qubits in $\ell$. By Eq.~\ref{eq:check_commute}, we can arbitrarily pair up the elements in $c^{Y, Z} \cap \ell$. Introduce a new edge for each such pair $qq'$. The edges will be data qubits in the ancilla system. $c$ is deformed to have support in the $Z$-basis on each edge $qq'$, which is sufficient to ensure that the checks of the ancilla system commutes with the original checks. Additional edges can be added to ensure that the graph $\mathcal{G}$ has Cheeger constant $\geq 1$.

Next, take the parsimonious cone of the measurement graph. All new faces will become $Z$-checks. By exactness of the cell complex at degree $1$, there are no new logical qubits introduced within the ancilla system, and a counting of stabilizer generators and data qubits quickly shows that there are no other logical qubits introduced to the deformed code. As $\mathcal{G}$ is expanding, no existing logical operators can be cleaned to have weight lower than $d$ in the deformed code.

\subsubsection{Arbitrary Pauli operator in a stabilizer code}

To consider the most general case, it is convenient to introduce the notion of basis-changing by local Cliffords. Given an initial stabilizer code $D$, define application of a single-qubit Clifford operator $\mathfrak{C}$ to a qubit $i$ as a map between codes $D \rightarrow D'$  which changes the basis of each check $c$ incident to $i$ by $\mathfrak{C}$, i.e.
\[\mathcal{P}(c^{D'}) \restriction_i = \mathfrak{C}^\dagger \mathcal{P}(c^{D}) \restriction_i\mathfrak{C}.\]
The code $D'$ still has commuting checks, and has the same parameters as $D$, but the checks may act in different bases. This map between codes is a tool for the purposes of constructing the correct ancilla system, and is not intended to be performed as part of the protocol on a quantum computer.

Now let $\mathcal{P}(\ell)^{D}$ be an arbitrary logical Pauli operator acting on the qubits $\ell \in \dF_2^n$ in code $D$. By applying local Cliffords to each qubit in $\ell$, map to a code $D'$ so that 
\[\mathcal{P}(\ell)^{D'} = X(\ell)^{D'} = (\bigotimes_{i \in \ell}\mathfrak{C}_i^\dagger)\mathcal{P}(\ell)^{D}(\bigotimes_{i \in \ell}\mathfrak{C}_i),\]
i.e. each qubit in $\ell$ is now acted upon in the $X$ basis.

Construct the parsimonious surgery ancilla system as in App.~\ref{sec:X_type_in_stab}, and deform $D'$ accordingly to yield a code $E'$. Then apply the adjoint map between codes, $(\bigotimes_{i \in \ell}\mathfrak{C}_i) E' (\bigotimes_{i \in \ell}\mathfrak{C}_i^\dagger) = E$, such that the deformed code now matches the original code $D$ on the original checks and qubits. The connections between checks in the ancilla code and qubits in the original code are now toggled to be in the correct basis for parsimonious surgery. As the codes $E$ and $E'$ have the same parameters, there are no new logical qubits introduced during the measurement, and the code distance is still at least $d$.

\subsection{Universal adapters}

Universal adapters~\cite{swaroop2026universal} are a way to perform joint logical measurements across codeblocks, or within the same codeblock. Given $t$ Pauli logical representatives, each with weight at most $\Weight$, the scheme achieves joint measurement of the operators with qubit overhead $\mathcal{O}(t\Weight\log^3\Weight)$. When the representatives are disjoint on physical qubits, or overlap only sparsely, the scheme is guaranteed to preserve the LDPC property.

\begin{remark}
    Joint measurement can also be performed directly by using the gauging logical measurement scheme~\cite{williamson2024low}. Assuming the $t$ representatives are disjoint, the weight of the product of those representatives could be $\Weight' = t\Weight$ in the worst case. Therefore direct measurement using gauging logical measurement would have asymptotic overhead,
    \[\mathcal{O}(\Weight'\log^3\Weight') = \mathcal{O}(t\Weight\log^3t\Weight) \geq \mathcal{O}(t\Weight\log^3\Weight)\]
    and hence be more expensive than universal adapters.
\end{remark}

Universal adapters begin by constructing the measurement graphs for each of the $t$ logical representatives individually, see Sec.~\ref{sec:measurement-graph}, and ensuring they each have Cheeger constant at least $1$. The measurement graphs may not possess sparse cycle bases, so the Decongestion lemma~\cite{freedman2021building, hastings2021quantum} is applied to each individually, thickening the graph and cellulating cycles. Decongestion adds new vertices and edges to each graph, such that the total number of qubits used so far is $\mathcal{O}(t\Weight \log^3 \Weight)$, see Ref.~\cite[Thm.~28]{he2025extractors}.

Secondly, the \textit{repetition code adapter} is applied, attaching new vertices and edges to the measurement graphs in order to efficiently construct a larger, connected auxiliary graph across all logicals. For brevity we omit the finer details of the repetition code adapter, which can be found in~\cite[Sec.~III-IV]{swaroop2026universal}. The most salient facts are that:
\begin{enumerate}
    \item The repetition code adapter always exists, can be found efficiently, and does not incur any additional asymptotic overhead.
    \item The repetition code adapter is attached only to the measurement graphs, and not to the additional vertices or edges used for decongestion.
    \item Assuming the measurement graphs each have Cheeger constant $\geq 1$, the auxiliary graph has \textit{relative} expansion $\geq 1$, which is sufficient to preserve distance of the deformed code.
\end{enumerate}

As a consequence, the decongestion and thickening step can be replaced by taking the parsimonious cone, for each of the $t$ measurement graphs. Because an measurement graph $G$ is always a subgraph of the parsimonious cone $C$, the repetition code adapter is guaranteed to applicable to the ancilla systems constructed by parsimonious surgery, using facts 1 and 2 above. As the measurement graphs each have Cheeger constant $\geq 1$, the overall auxiliary graph still has relative expansion $\geq 1$, using fact 3.

Thus universal adapters with base graphs constructed by parsimonious surgery have space overhead $\mathcal{O}(t\Weight \log \Weight)$, lower than the prior best of $\mathcal{O}(t\Weight \log^3 \Weight)$.

\subsection{Parallel Logical Measurements}\label{app:parallel_measurement}

When logical representatives overlap substantially in support on physical qubits, fault-tolerantly measuring them simultaneously (but not jointly) becomes difficult, as the ancilla systems to measure each one can overlap substantially, increasing the density of the deformed code and, in the worst case, preventing the stabilizers of the deformed code from commuting. As high logical representative overlap is an intrinsic property of high rate and distance LDPC codes, these can be substantial problems for low spacetime quantum computation with LDPC codes.

To ameliorate these problems, schemes for parallel logical measurement have been devised~\cite{zhang2025time,cowtan2025parallel,zheng2025high}. Here we focus on Ref.~\cite{cowtan2025parallel}, a general scheme with the lowest asymptotic space overhead to date.

The scheme proceeds as follows. First, we require the following definition (cf.~\cite[Def.~2.10]{cowtan2025parallel}):

\begin{definition}
    Two logical representatives $\Gamma$ and $\Lambda$ are said to qubit-wise commute if, for every qubit $q$ such that $q\in {\rm supp}(\Gamma)$ and $q \in {\rm supp}(\Lambda)$,
    \[[\Gamma \restriction_q, \Lambda \restriction_q] = 0\]
    where $A \restriction_b$ is the restriction of the Pauli action of operator $A$ to qubit $b$.
\end{definition}

For measurement, an ancilla system is attached in an intermediate step such that new representatives are introduced for each Pauli logical $\Lambda_i$ to be measured, via a procedure called \textit{branching}~\cite{zhang2025time}. These new representatives $\Lambda_i'$ are stabilizer equivalent to the original $\Lambda_i$ representatives, and the representatives $\Lambda_i$ are said to have been \textit{branched}. The ancilla system commutes when the representatives to be branched qubit-wise commute on physical qubits.

Finally, for each $\Lambda_i'$, graphs are attached using the gauging logical measurement formalism~\cite{williamson2024low}. As with universal adapters, we skim over the finer details of the protocol, but present the most salient facts below:

\begin{enumerate}
    \item A logical basis always exists for any stabilizer code such that, for any set of logical representatives acting on different \textit{logical} qubits, the representatives qubit-wise commute~\cite[Sec.~3.4]{cowtan2025parallel} on \textit{physical} qubits.
    \item If the initial code is LDPC, the system after branching remains LDPC, and each representative $\Lambda_i'$ to be measured is disjoint on physical qubits. Each representative $\Lambda_i'$ has weight identical to that of $\Lambda_i$.
    \item If the initial code has distance $d$, the distance of the deformed code after branching has distance at least $d$, and the branching protocol has phenomenological fault-distance at least $d$.
    \item Given $t$ initial logical representatives with weight at most $\Weight$, branching uses $\mathcal{O}(t\Weight\log t)$ new qubits.
\end{enumerate}

Applying gauging logical measurement for each representative $\Lambda_i'$ after branching has a cost of $\mathcal{O}(t\Weight\log^3\Weight)$, where the $\log^3\Weight$ contribution is due to the Decongestion lemma~\cite{hastings2021quantum}, thickening and cellulation. Thus the total previous cost of the parallel logical measurement scheme was 
\[\mathcal{O}(t\Weight\log t + t\Weight\log^3\Weight) = \mathcal{O}(t\Weight(\log t + \log^3 \Weight)).\]
Using such a scheme, any set of logical operators acting on disjoint logical qubits in any stabilizer code can be measured in parallel, presuming that the logical basis is given by the specific qubit-wise commuting one~\cite{cowtan2025parallel}.

By applying parsimonious surgery instead of gauging logical measurement, the $\log^3\Weight$ contribution is reduced to $\log \Weight$. As a consequence, the overall asymptotic qubit cost of the protocol becomes
\[\mathcal{O}(t\Weight(\log t + \log \Weight)).\]

\subsection{Extractor Architectures}

An extractor architecture~\cite{he2025extractors} is a general model that describes fault-tolerant universal quantum computers which implement Pauli-based computation~\cite{bravyi2016trading} on qLDPC codes. 
The model's broad applicability come from the fact that extractor architectures can be designed with arbitrary qLDPC codes, can be implemented on both fixed-connectivity and reconfigurable physical qubit platforms, and host a vast design space for application- and hardware-specific optimizations. 
Practical instantiations of extractor architectures~\cite{yoder2025tour,webster2026pinnacle} have been proposed and studied.

The key primitive used in extractor architectures is called an \textit{extractor system}~\cite{he2025extractors}.
Each codeblock in the architecture is augmented with a \textit{single-block extractor system}, which is a simply connected graph with a design similar to that of gauging logical measurement. For each stabilizer generator of the codeblock, a cycle is conferred to the graph, and the graph is designed to have Cheeger constant $\geq 1$ and to have a low-weight cycle basis. The graph is connected to the original codeblock via a \textit{port function}, such that any single logical Pauli operator on the codeblock can be measured in one logical timestep by modifying the connections from the codeblock to the measurement graph; this modification requires only deleting edges from the Tanner graph, and so can be accomplished with fixed connectivity.

Single-block extractor systems are then connected via bridges, which are similar to repetition code adapters~\cite{swaroop2026universal}, to create \textit{multi-block extractors}. Product measurements across multiple codeblocks are accomplished by activating the \textit{bridge qubits}.

\begin{remark}
    In practice, many codes possess symmetry properties allowing for implementation of some logical Clifford gates~\cite{breuckmann2024fold,eberhardt2024logical}. Therefore it is often sufficient to use \textit{partial extractors}, which do not connect to every logical representative in the original codeblock, and can thus be practically optimised to reduce the space cost~\cite{yoder2025tour, webster2026pinnacle}.
\end{remark}

The most salient facts about extractor systems are the following:
\begin{enumerate}
    \item For a codeblock with length $n$, a single-block extractor has asymptotic space cost $\mathcal{O}(n\log^3 n)$, via the familiar Decongestion lemma, thickening and cellulation. A suitable single-block extractor exists for every codeblock.
    \item The port function connects the original codeblock only to the original measurement graph before thickening, and not to the additional vertices or edges used for decongestion.
    \item The bridge system always exists, can be found efficiently, and does not incur any additional asymptotic overhead.
    \item The bridge system is attached only to the measurement graph, and not to the additional vertices or edges used for decongestion.
\end{enumerate}

Due to facts 2 and 4, the decongestion and thickening step for the single-block extractor can be replaced by taking the parsimonious cone, without modifying the scheme elsewhere. The connections to the original code, and to other single-block extractors via the bridge system, remain unchanged under this replacement.

As a consequence, the asymptotic space overhead of extractors is reduced from the prior best of $\mathcal{O}(n\log^3 n)$ to $\mathcal{O}(n\log n)$.

\subsection{Single-shot surgery}\label{app:single-shot}

Fast surgery with qLDPC codes, which takes $o(d)$ rounds of syndrome extraction to fault-tolerantly measure one or more logical operators, has been a focus of recent study~\cite{chang2026constant, hillmann2025single, cowtan2025fast,baspin2025fast}.

In this section we discuss the single-shot surgery gadget constructed in Ref.~\cite[Sec.~IV B]{baspin2025fast}. 
The goal is to measure an $X$-type or $Z$-type operator from a CSS LDPC code.
The protocol starts by taking a perfect round of syndrome measurement on the base code, then uses an auxiliary system to measure a logical operator for $\mathcal{O}(1)$ rounds, and end with another perfect round of syndrome measurement on the base code. 
By introducing sufficient meta-checks (into the ancilla system) to protect the new ancilla stabilizers, which determine the logical measurement outcome, from error, Ref.~\cite{baspin2025fast} proves that this single-shot surgery protocol is fault-tolerant. 
Because of the assumption of perfect syndrome rounds before and after the surgery measurement,
this protocol is fast for codes which admit single-shot state preparation. 
For more generic codes, obtaining perfect rounds of syndromes may take $\Omega(d)$ rounds of stabilizer measurements.\footnote{See Ref.~\cite[Sec.~IA]{chang2026constant} for further discussion.}

The ancilla system is constructed in the following manner:
\begin{enumerate}
    \item First, a copy $A$ of the original code $D$ is taken as a starting point for the auxiliary chain complex, with the identity chain map $f = I$ between them. This uses $\mathcal{O}(n)$ qubits.
    \item Second, each logical operator in $A$ \textit{apart from} the operator desired to be measured in a single-shot fashion is converted into a stabilizer using the parallel logical measurement scheme of Ref.~\cite{cowtan2025parallel}, which adds up to $\mathcal{O}(nk(\log k+\log^3 n))$ new qubits to the auxiliary system.
    \item Third, $A$ is thickened up to $d$ times, bringing the total qubit cost to $\mathcal{O}(nkd(\log k+\log^3 n))$, in order to boost the expansion of the ancilla system such that the distance of the deformed code is high.
    \item Fourth, take ${\rm cone}(f)$, degree-shifting $A$ such that what was previously the qubit degree is now a check degree, to create a deformed code with $k-1$ logical operators.
\end{enumerate}

Crucially, due to the large number of checks in the ancillary system and structural properties of the resultant chain complex, the deformed code has sufficient metachecks such that the logical measurement can be performed in $\mathcal{O}(1)$ rounds.

Because the parallel logical measurement scheme of Ref.~\cite{cowtan2025parallel} is used as a subroutine, the knock-on effect of taking parsimonious cones as opposed to decongestion and thickening, described in App.~\ref{app:parallel_measurement}, can be immediately applied to the single-shot surgery scheme as well, without affecting the structural properties that allow for fast measurement. Hence the space cost of the single-shot surgery scheme is reduced from $\mathcal{O}(nkd(\log k+\log^3 n))$ to $\mathcal{O}(nkd(\log k+\log n))$.

\subsection{Constant-time surgery}

A different formulation of single-shot surgery is called constant-time surgery, which was introduced in Ref.~\cite{cowtan2025fast}. 
The gist is that instead of performing one surgery operation in $O(1)$ time, we perform $O(d)$ operations in $O(d)$ time, so that every surgery operation costs $O(1)$ time in amortization.
This protocol is fault-tolerant given $O(d)$ rounds of syndrome measurement before and after the $O(d)$ surgery operations.
Importantly, it does not require the base code to admit single-shot state preparation. Its notion of fault-tolerance is broadly analogous to that of algorithmic fault-tolerance~\cite{zhou2025low}. 

The gadgets constructed in Ref.~\cite{chang2026constant} on 2D hypergraph product codes follow this constant-time surgery formulation. 
Each gadget measures a collection of logical $Z$-type (or $X$-type) operators in amortized $\mathcal{O}(1)$ time. 
As with the other schemes described in this appendix, for brevity we avoid reiterating the homological algebra used to construct the protocol. At a high level, the construction uses the following steps:
\interfootnotelinepenalty=10000
\begin{enumerate}
    \item Certain sets of logical operators in a hypergraph product code can each be described by a product of (a) a classical codeword $\ell$ and (b) the physical space $D$ of a classical code.\footnote{In the perspective of hypergraph product logical qubits in a grid~\cite{xu2025fast}, this corresponds to a row from the grid.} Choose such a set $\mathcal{L}$; the classical codeword and the physical space each have size at most $\mathcal{O}(\sqrt{n})$.
    \item Construct the gauging logical measurement graph $\mathcal{G}$ for $\ell$, viewed as a logical operator of a quantum code with no $Z$ stabilizers. This uses $\mathcal{O}(\sqrt{n}\log^3\sqrt{n}) = \mathcal{O}(\sqrt{n}\log^3 n)$ space, by decongestion and thickening.
    \item Take the tensor product $\mathcal{G}\otimes D$ to form a length-2 chain complex. This complex always has sufficient expansion and possesses a canonical chain map $f$ into the original code, determined by the codeword $\ell$ and space $D$, and has size at most $\mathcal{O}(\sqrt{n} \circ \sqrt{n} \log^3 n) = \mathcal{O}(n \log^3 n)$.
    \item The deformed code is given by ${\rm cone}(f)$, and has used $\mathcal{O}(n \log^3 n)$ new qubits to measured $k_D$ logical operators simultaneously, in $\mathcal{O}(1)$ rounds.
\end{enumerate}

The reason why a set of such measurements can be performed in constant time in amortization is that each measurement takes $\mathcal{O}(1)$ rounds due to a sufficient set of metachecks,
and furthermore it is proven that there are no low-weight logical faults which can extend between measurement rounds, even when the idling time between measurement rounds is $\mathcal{O}(1)$~\cite{cowtan2025fast}.

In the protocol above, the algebraic properties such as the canonical chain map and sufficient metacheck distance, are invariant upon exchanging the gauging logical measurement for parsimonious surgery. As a consequence, that step can be replaced and the space overhead reduced to $\mathcal{O}(n \log n)$.










\end{document}